\documentclass[12pt, letterpaper]{article}
\usepackage{geometry}
\usepackage{natbib}
\usepackage{amsmath,amssymb,amsfonts,amsthm,bbm}
\usepackage{epic,eepic,epsfig,longtable}
\usepackage{multirow,verbatim}
\usepackage{esvect}
\usepackage{array}
\usepackage{algorithm}
\usepackage{multirow}
\usepackage{footnote}
\usepackage{threeparttable} 
\usepackage{epsfig}
\usepackage{setspace}
\usepackage{enumitem}
\usepackage{color}
\usepackage{fancyhdr}
\usepackage{appendix}
\usepackage{listings}
\usepackage{longtable}
\usepackage{url}
\usepackage{lineno,hyperref}
\usepackage{subfigure}
\usepackage{graphicx}
\usepackage{booktabs}
\usepackage{longtable}
\usepackage{siunitx}
\usepackage{amsthm,mathbbol,bbm}
\usepackage{float}
\usepackage{multirow}
\usepackage{epsfig}
\usepackage{setspace}
\usepackage{color}
\usepackage{fancyhdr}
\usepackage{appendix}
\usepackage{listings}
\usepackage{longtable}
\usepackage{url}
\usepackage{lineno,hyperref}
\usepackage{subfigure}
\usepackage{mathrsfs}
\usepackage{stmaryrd}
\usepackage{appendix}
\usepackage{algorithm}
\usepackage{enumitem}
\usepackage{algpseudocode}
\usepackage{caption}
\usepackage{amsmath}
\usepackage{graphics}
\usepackage{epsfig}
\usepackage{footnote}
\usepackage{threeparttable} 

\algdef{SE}[DOWHILE]{Do}{doWhile}{\algorithmicdo}[1]{\algorithmicwhile\ #1}%

\def\c{\centerline}

\def\bzero{\bfm 0}

\newcommand{\bfsym}[1]{\ensuremath{\boldsymbol{#1}}}

\def\balpha{\bfsym \alpha}
\def\bbeta{\bfsym \beta}
\def\bgamma{\bfsym \gamma}             
\def\bDelta {\bfsym {\Delta}}
\def\bfeta{\bfsym {\eta}}

\def\btheta{\bfsym {\theta}}           
          \def\bepsilon{\bfsym \varepsilon}
\def\bJ{\bfsym \sigma}             \def\bJ{\bfsym \Sigma}

\def\bpsi{\bfsym {\psi}}		        \def\bPsi{\bfsym {\Psi}}

\def\bzeta{\bfsym {\zeta}}

\def\today{\ifcase\month\or
	January\or February\or March\or April\or May\or June\or
	July\or August\or September\or October\or November\or December\fi
	\space\number\day, \number\year}

\def\bb{\boldsymbol}
 at 8truept

\def\U{\mbox{\smallfont{U}}}

\def\W{\boldsymbol{W}}

\def\w{{\bf w}}

\newcommand{\beq}{\begin{equation}\tag{1}}
\newcommand{\eeq}{\end{equation}}
\newcommand{\beqn}{\begin{eqnarray}}
\newcommand{\eeqn}{\end{eqnarray}}
\newcommand{\beqnn}{\begin{eqnarray*}}
\newcommand{\eeqnn}{\end{eqnarray*}}

\def\bbw{\mathbbm{w}}

\def\U{{\cal{U}}}

\def\Psc{\mathcal{P}}
\def\Esc{\mathcal{E}}
\def\bC{\boldsymbol{C}}
\def\bdelta{\boldsymbol{\delta}}
\def\a{\boldsymbol{a}}
\def\b{\boldsymbol{b}}

\def\subminusk{_{\scriptscriptstyle \text{-}k}}

\def\subDR{_{\scriptscriptstyle \sf DR}}
\def\subIW{_{\scriptscriptstyle \sf IW}}

\def\Xtilde{\widetilde{X}}
\def\bXtilde{\widetilde{\X}}

\allowdisplaybreaks
\usepackage{verbatim}
\newcounter{CondCounter}

\def\Tsc{\mathcal{T}}
\def\Ssc{\mathcal{S}}
\def\trans{^{\scriptscriptstyle \sf T}}
\def\X{\boldsymbol{X}}
\def\x{\boldsymbol{x}}
\def\c{\boldsymbol{c}}

\newtheorem{lemma}{Lemma}
\newtheorem{assumption}{Assumption}

\newtheorem{theorem}{Theorem}
\newtheorem{remark}{Remark}
\newtheorem{proposition}{Proposition}

\def\bbeta{\boldsymbol{\beta}}
\def\balpha{\boldsymbol{\alpha}}
\def\bgamma{\boldsymbol{\gamma}}
\def\btheta{\boldsymbol{\theta}}
\def\balphahat{\widehat{\balpha}}
\def\bgammahat{\widehat{\bgamma}}
\def\omegahat{\widehat{\omega}}
\def\bbetahat{\widehat{\bbeta}}

\def\mhat{\widehat{m}}
\def\bphi{\boldsymbol{\phi}}
\def\bpsi{\boldsymbol{\psi}}
\def\bJ{\boldsymbol{J}}
\def\bPi{\boldsymbol{\Pi}}
\def\bfeta{\boldsymbol{\eta}}

\def\Isc{\mathcal{I}}
\def\supfk{^{[\text{-}k]}}
\def\supfktrans{^{[\text{-}k]^{\scriptscriptstyle \sf T}}}
\def\V{\boldsymbol{V}}

\def\bJhat{\widehat{\bJ}}

\def\bzeta{\boldsymbol{\zeta}}

\def\z{\boldsymbol{z}}

\def\Z{\boldsymbol{Z}}
\def\bPhi{\boldsymbol{\Phi}}
\def\bPsi{\boldsymbol{\Psi}}
\def\balphabar{\bar{\balpha}}

\def\rhat{\widehat{r}}

\def\c{\boldsymbol{c}}
\def\U{\boldsymbol{U}}
\def\bUhat{\widehat{\U}}

\def\subatrel{_{\scriptscriptstyle \sf ATReL}}
\def\subvalid{_{\scriptscriptstyle \sf Valid}}
\def\mbar{\bar{m}}

\def\bxi{\boldsymbol{\xi}}

\definecolor{darkred}{RGB}{150,50,50}
\definecolor{brown}{RGB}{250,100,100}
\definecolor{green}{RGB}{000,150,100}
\definecolor{purple}{RGB}{250,000,180}

\def\bzeta{\boldsymbol{\zeta}}
\def\A{\boldsymbol{A}}
\def\hbar{\bar{h}}
\def\rbar{\bar{r}}
\def\bgammabar{\bar{\bgamma}}
\def\hhat{\widehat{h}}
\def\rP{\mathbb{P}}
\def\rE{\mathbb{E}}
\def\Lsc{\mathcal{L}}
\def\Nsc{\mathcal{N}}
\def\Zsc{\mathcal{Z}}
\def\Csc{\mathcal{C}}

\title{Augmented Transfer Regression Learning with Semi-non-parametric Nuisance Models}

\author{Molei Liu$^{\dagger*}$, Yi Zhang$^{\ddagger}$, Katherine P Liao$^{\S}$, Tianxi Cai$^{\dagger}$}

\footnotetext[1]{The first two authors are the joint first authors.}
\footnotetext[2]{Department of Biostatistics, Harvard Chan School of Public Health.}
\footnotetext[3]{Department of Statistics, Harvard University.}
\footnotetext[4]{Department of Medicine Rheumatology, Immunology, Brigham and Women's Hospital.}
\footnotetext[5]{The is the second version of our manuscripts. The first version was on arxiv Oct 2020.}

\begin{document}
\maketitle
\begin{abstract}
{
\noindent In contemporary statistical learning, covariate shift correction plays an important role in transfer learning when distribution of the testing data is shifted from the training data. Importance weighting \citep{huang2007correcting}, as a natural and principle strategy to adjust for covariate shift, has been commonly used in the field of transfer learning. However, this strategy is not robust to model misspecification or excessive estimation error. In this paper, we propose an augmented transfer regression learning (ATReL) approach that introduces an imputation model for the targeted response, and uses it to augment the importance weighting equation. With novel semi-non-parametric constructions and calibrated moment estimating equations for the two nuisance models, our ATReL method is less prone to (i) the curse of dimensionality compared to  nonparametric approaches, and (ii) model mis-specification than parametric approaches.  We show that our ATReL estimator is $n^{1/2}$-consistent when at least one nuisance model is correctly specified, estimation for the parametric part of the nuisance models achieves parametric rate, and the nonparametric components are rate doubly robust. Simulation studies demonstrate that our method is more robust and efficient than existing parametric and fully nonparametric (machine learning) estimators under various configurations. We also examine the utility of our method through a real example about transfer learning of phenotyping algorithm for rheumatoid arthritis across different time windows. Finally, we propose ways to enhance the intrinsic efficiency of our estimator and to incorporate modern machine learning methods with our proposed framework. 
}

\end{abstract}


\newpage
\section{Introduction}\label{sec:intro}

\subsection{Background}\label{sec:intro:back}

The shift in the predictor distribution, often referred to as {\em covariate shift}, is one of the key contributors to poor transportability and generalizability of a supervised learning model from one data set to another. An example that arises often in modern biomedical research is the between health system transportability of prediction algorithms trained from electronic health records (EHR) data \citep{weng2020deep}. Frequently encountered heterogeneity between hospital systems include the underlying patient population and how the EHR system encodes the data. For example, the prevalence of rheumatoid arthritis (RA) among patients with at least one billing code of RA differ greatly among hospitals \citep{carroll2012portability}. 
On the other hand, the conditional distribution of the disease outcome given all important EHR features may remain stable and similar for different cohorts. Nevertheless, shift in the distribution of these features can still have a large impact on the performance of a prediction algorithms trained in one source cohort on another target cohort \citep{rasmy2018study}. Thus, correcting for the covariate shift is crucial to the successful transfer learning across multiple heterogeneous studying cohorts.

Robustness of covariate shift correction is an important topic and has been widely studied in recent literature of statistical learning. A branch of work including \cite{wen2014robust,chen2016robust,reddi2015doubly,liu2017robust} focused on the covariate shift correction methods that are robust to the extreme importance weight incurred by the high dimensionality. Main concern of their work is the robustness of a learning model's prediction performance on the target data to a small amount of high magnitude importance weight. However, there is a paucity of literature on improving the validity and efficiency of statistical inference under covariate shift, with respect to the robustness to the mis-specification or poor estimation of the importance weight model. In this paper, we propose an augmented transfer regression learning (ATReL)  procedure in the context of covariate shift by specifying flexible machine learning models for the importance weight model and the outcome model. We establish the validity and efficiency of the proposed method under possible mis-specification in one of the specified models. We next state the problem of interest and then highlight the contributions of this paper.

\subsection{Problem Statement}\label{sec:intro:prob}

The source data, indexed by $S=1$, consist of $n$ labeled samples with observed response $Y$ and covariates $\X=(X_1,\ldots,X_p)$ while the target data, indexed by $S=0$, consist of $N$ unlabeled samples with only observed on $\X$. We write the full observed data as $\{(S_iY_i,\X_i,S_i):i=1,2,\ldots,n+N\}$, where without loss of generality we let the first $n$ observations be from the source population with $S_i = I(1 \le i \le n)$ and remaining from the target population. We assume that 
 $(Y,\X) \mid S=s\sim p_{s}(\x)q(y\mid\x)$, where $p_s(\x)$ denotes the probability density measure of $\X \mid S=s$ and $q(y\mid\x)$ is the conditional density of $Y$ given $\X$, which is the same across the two populations. The conditional distribution of $Y \mid \X$, shared between the two populations, could be complex and difficult to specify correctly. {In practice, it is often of interest to infer about a functional of $\mu(\X)$ such as $\rE(Y \mid \A, S=0)$, where $\A\in\mathbb{R}^d$ is a sub-vector of $\X$. More generally, we consider a working model $\rE_0(Y \mid \A)=g(\A\trans\bbeta)$ and define the regression parameter $\bbeta_0$ as the solution to the estimating equation in the target population $S=0$:
 \begin{equation}
     \rE\left[\A\{Y-g(\A\trans\bbeta)\} \mid S=0\right] \equiv \rE_{0}[\A\{Y-g(\A\trans\bbeta)\}]=\mathbf{0},\label{def-beta}
 \end{equation}
 where $\rE_{s}$ is the expectation operator on the population $S = s$ and $g(\cdot)$ is a link function, e.g. $g(\theta)=\theta$ represents linear regression and $g(\theta)=1/(1+e^{-\theta})$ for logistics regression. Directly solving an empirical estimating equation for (\ref{def-beta}) using the source data to estimate $\bbeta_0$ may result in inconsistency due to the covariate shift as well as potential model mis-specification of the model $\rE_0(Y \mid \A)=g(\A\trans\bbeta)$. It is important to note that even when $\rE_0(Y \mid \A) = g(\A\trans\bbeta_0)$ holds, $\rE_1\{\A(Y-g(\A\trans\bbeta_0)\}$ may not be zero in the presence of covariate shift.} To correct for the covariate shift bias, it is natural to incorporate importance sampling weighting and estimate $\bbeta_0$ as $\bbetahat\subIW$, the solution to the weighted estimating equation
\begin{equation}
\frac{1}{n}\sum_{i=1}^n\widehat{\omega}(\X_i)\A_i\{Y_i-g(\A_i\trans\bbeta)\}=0,
\label{equ:iw}
\end{equation}
where $\omegahat(\X)$ is an estimate for the density ratio $\bbw(\X)=p_0(\X)/p_1(\X)$. However, the validity of $\bbetahat\subIW$ heavily relies on the consistency of $\widehat{\omega}(\X)$ for $\bbw(\X)$ and can perform poorly when the density ratio model is mis-specified or not well estimated. 


{
\begin{remark}
Our goal is to infer the conditional model of $Y$ on $\A$, a low dimensional subset of covariates in $\X$. 
In practice, there are a number of such cases in which one would be interested in a ``submodel" $Y\sim \A$ rather than the ``full model" $Y\sim \X$. For example, in EHR studies, $\A$ may represent widely available codified features and other elements of $\X$ may include features extracted from narrative notes via naturally language processing (NLP), which can be available for research studies but too costly to include when implementing risk models for broad patient populations. 
Also, when predicting the risk of developing a future event $Y$ at baseline, $\A$ may represent baseline covariates while the remaining elements of $\X$ may include post baseline surrogate features that can be used to``impute" $Y$ but not meaningful as risk factors.
\label{rem:target}
\end{remark}}
In this paper, we propose an augmented transfer regression learning (ATReL) method for optimizing the estimation of a potentially mis-specified regression model. Building on top of the augmentation method in the missing data literature, our method leverages a flexible semi-non-parametric outcome model $m(\X)$ imputing the missing $Y$ for the target data and augments the importance sampling weighted estimating equation with the imputed data. It is doubly robust (DR) in the sense that the ATReL estimator approaches the target $\bbeta_0$ when either the importance weight model $\omega(\X)$ or the imputation model $m(\X)$ is correctly specified.

\subsection{Literature review and our contribution}\label{sec:intro:lit}

Doubly robust estimators have been extensively studied for missing data and causal inference problems \citep{bang2005doubly,qin2008efficient,cao2009improving,van2010collaborative,tan2010bounded,vermeulen2015bias}. Estimation of average treatment effect on the treated can be viewed as analog to our covariate shift problem. To improve the DR estimation for average treatment effect on the treated, \cite{graham2016efficient} proposed a auxiliary-to-study tilting method and studied its efficiency. \cite{zhao2017entropy} proposed an entropy balancing approach that achieves double robustness without augmentation and \cite{shu2018improved} proposed a DR estimator attaining local and intrinsic efficiency. Besides, existing work like \cite{rotnitzky2012improved} and \cite{han2016intrinsic} are similar to us in the sense that their parameters of interests are multidimensional regression coefficients. Properties including intrinsic efficiency and multiple robustness has been studied in their work. These methods used low dimensional parametric nuisance models in their constructions, which is prone to bias due to model mis-specification.



To improve robustness to model mis-specifications,  \cite{rothe2015semiparametric} used local polynomial regression to estimate the nuisance functions in constructing the DR estimator for an average treatment effect. 
\cite{chernozhukov2016double} extended classic nonparametric constructions to the modern machine learning setting with cross-fitting. Their proposed double machine learning (DML) framework facilitates the use of general machine learning methods in semiparametric estimation. This general framework has also been explored for semiparametric models with non-linear link functions \citep[e.g.]{semenova2020debiased,liu2020note}. In contrast to the parametric approaches, the fully nonparametric strategy is free of mis-specification of the nuisance models. However, it is impacted by the excessive fitting errors of nonparametric models with higher complexity than parametric models, and thus subject to the so called ``rate double robustness" assumption \citep{smucler2019unifying}. Typically, classic nonparametric regression methods like kernel smoothing could not achieve the desirable convergence rates even under a moderate dimensionality. Though such ``curse of dimensionality" could be relieved by modern machine learning methods like random forest and neural network, theoretical justification on the performance of these methods are inadequate. Even their asymptotic convergence are sometimes justifiable, these machine learning approaches still requires particularly large sample sizes to ensure good finite sample performances, which could be seen from our numerical studies. This drawback has became a main concern about the nonparmatric or machine learning approaches.


{Our proposed semi-non-parametric strategy in constructing the nuisance models can be viewed as a mitigation of the parametric and nonparamertic methods, which is more flexible and powerful. In specific, it specifies the two nuisance models as the generalized partially linear models combining a parametric function of some features in $\X$ and a nonparametric function of the other features, to achieve a better trade-off in model complexity. It is more robust to model estimation errors compared to the fully nonparametric approach, and less susceptible to model mis-specification than the parametric approach. Our method is not a trivial extension of the two existing strategies as we construct the moment equations more elaborately to \emph{calibrate} the nuisance models, and remove the over-fitting bias. We take semi-non-parametric models with kernel or sieve estimator as our main example for realizing this strategy, and present other possibilities including the high dimensional regression and machine learning constructions. We show that the proposed estimator is $n^{1/2}$-consistent and asymptotically normal when at least one nuisance model is correctly specified, the parametric components in the two models are $n^{1/2}$-consistent, and both nonparametric components attain the error rate $o_p(n^{-1/4})$.

}


In existing literature of semiparametric inference, one alternative and natural way to mitigate the model misspecification and the curse of dimensionality is to construct the nuisance models with some high dimensional non-linear basis of $\X$. In relation to this, a number of recent works has been developed to construct model doubly robust estimators using high dimensional sparse nuisance models \citep[e.g.]{smucler2019unifying,tan2020model,ning2020robust,dukes2020inference,ghosh2020doubly,liu2020note}. The central idea of these approaches is to impose certain moment conditions on the nuisance models to remove their first order (or over-fitting) bias under potential model misspecification, which is referred as calibrating \citep{tan2020model}. Technically, our calibrating procedure is in similar spirits with this idea. Different from their strategies to fit regularized high dimensional regression with all covariates, we treat the parametric and the nonparametric parts in the nuisance model differently. And our parametric part can be specified by arbitrary estimating equations. This provides us more flexibility on model specification, as well as possibility to achieve better intrinsic efficiency as discussed in Section \ref{sec:disc}. {More importantly, our framework allows for the use of nonparmatric or machine learning methods like kernel smoothing and random forest, while these existing methods are restricted to high dimensional parametric models. In addition, our target is a regression model, which has larger complexity than the single average treatment effect parameter studied in the previous work, and incurs additional challenges like irregular weights.}

A similar idea of constructing semi-non-parametric nuisance models has been considered by \cite{chakrabortty2016robust} and \cite{chakrabortty2018efficient} using this to improve the efficiency of linear regression under a semi-supervised setting with no covariate shift between the labeled and unlabeled data. They proposed a refitting procedure to adjust for the bias incurred by the nonparametric components in the imputation model while our method can be viewed as their extension leveraging the importance weight and imputation models to correct for the bias of each other, which is substantially novel and more challenging. As another main difference, we use semi-non-parametric model in estimating the parametric parts of the nuisance models, to ensure their correctness and validity. \cite{chakrabortty2016robust} and \cite{chakrabortty2018efficient} did not actually elaborate on this point and only used parametric regression to estimate the parametric part, which does not guarantee the model double robustness property achieved by our method.

\subsection{Outline of the paper}
Remaining of the paper will be organized as follow. In Section \ref{sec:method}, we introduce the general doubly robust estimating equation, our semi-non-parametric framework and specific procedures to estimate the parametric and nonparametric components of nuisance models. In Section \ref{sec:thm}, we present the large sample properties of our proposed ATReL estimator, i.e. its double robustness concerning model specification and estimation. In Section \ref{sec:simu}, we present simulation results evaluating the finite sample performance of our ATReL estimator and its relevant performance compared with existing methods under various settings. In Section \ref{sec:real}, we apply our ATReL estimation on transferring a phenotyping algorithm for bipolar disorder across two EHR cohorts. Finally, we propose and comment on some potential strategies for improving and extending our method in Section \ref{sec:disc}.

\section{Method}\label{sec:method}

\subsection{General form of the doubly robust estimating equation}

Let $m(\x)$ denote an imputation model used to approximate $\mu(\x)=\rE(Y|\X=\x)=\rE_0(Y|\X=\x)=\rE_1(Y|\X=\x)$, and $\mhat(\x)$ denote the estimate of $m(\x)$ by fitting the model to the labeled source data. We augment the importance sampling weighted estimating equation (\ref{equ:iw}) with the term
\begin{equation}
\frac{1}{N}\sum_{i=n+1}^{N+n}\A_i\{\mhat(\X_i)-g(\A_i\trans\bbeta)\}-\frac{1}{n}\sum_{i=1}^n\widehat{\omega}(\X_i)\A_i\{\mhat(\X_i)-g(\A_i\trans\bbeta)\},
\label{equ:add}
\end{equation}
which results in the augmented estimating equation:
\begin{equation}
\bUhat\subDR(\bbeta)\equiv\frac{1}{n}\sum_{i=1}^n\widehat{\omega}(\X_i)\A_i\{Y_i-\mhat(\X_i)\}+\frac{1}{N}\sum_{i=n+1}^{N+n}\A_i\{\mhat(\X_i)-g(\A_i\trans\bbeta)\}=\mathbf{0}.
\label{equ:dr}
\end{equation}
We denote its solution as $\bbetahat\subDR$. {Construction (\ref{equ:dr}) is in the similar spirit with the DR estimators of the average treatment effect on the treated studied in existing literature \citep[e.g.]{graham2016efficient,shu2018improved}.} When the density ratio model is correctly specified and consistently estimated, equation (\ref{equ:dr}) converges to $\rE_0[\A_i(Y_i - g(\A_i\trans\bbeta)\}]=0$ and hence $\widehat\bbeta\subDR$ is consistent for $\bbeta_0$. When the imputation model is correct, the first term of $\bUhat\subDR(\bbeta)$ in (\ref{equ:dr}) converges to $\bzero$ and the second term converges to $\rE_0[\A_i\{E_0(Y_i\mid\X_i)-g(\A_i\trans\bbeta)\}]=\rE_0[\A_i\{Y_i-g(\A_i\trans\bbeta)\}]$ and hence $\bbetahat\subDR$ is also expected to be consistent  for $\bbeta_0$. Thus, the augmented estimating equation (\ref{equ:dr}) is doubly robust to the specification of the two nuisance models.

\subsection{Semi-non-parametric nuisance models}\label{sec:method:mom}

Now we introduce a semi-non-parametric construction for the nuisance models in (\ref{equ:dr}) that captures more complex effects in $\bbw(\X)$ and $\mu(\X)$ from a subset of $\X$, denoted by $\Z\in\mathbb{R}^{p_{\z}}$, along with simpler effects for the remainder of $\X$ that can be explained via linear effects on a finite set of pre-specified functional bases for approximating $\bbw(\X)$ and $\mu(\X)$, respectively denoted by $\bpsi\in\mathbb{R}^{p_{\bpsi}}$ and $\bphi\in\mathbb{R}^{p_{\bphi}}$. In EHR data analysis, $\Z$ may represent measures of healthcare utilization which may differ greatly across healthcare systems and have complex effects on patient outcome. Under this framework, we specify the following semi-non-parametric nuisance models for $\bbw(\X)$ and $\mu(\X)$,
\begin{equation}
\omega(\X)=\exp\{\bpsi\trans\balpha+h(\Z)\}\quad\mbox{and}\quad m(\X)=g\{\bphi\trans\bgamma+r(\Z)\},
\label{model}
\end{equation}
where $\bpsi\trans\balpha$ and $\bphi\trans\bgamma$ represent parametric components, the unknown functions $h(\z)$ and $r(\z)$ represent the nonparametric components, and $g(\cdot)$ is a pre-specified smooth strictly increasing link function. Without loss of generality, let the first element in both $\bpsi$ and $\bphi$ be constant $1$. Correspondingly, we denote their estimation used in (\ref{equ:dr}) as ${\omegahat}(\X)=\exp\{\bpsi\trans\balphahat+\widehat{h}(\Z)\}$ and $\mhat(\X)=g\{\bphi\trans\widehat{\bgamma}+\widehat{r}(\Z)\}$. Here and in the sequel, we let $\bbetahat\subatrel$ denote the ATReL estimator derived from (\ref{equ:dr}) with this specific construction of $\mhat(\cdot)$ and $\omegahat(\cdot)$.

\def\nhalf{n^{\half}}
\def\nnhalf{n^{-\half}}
\def\half{\frac{1}{2}}
\def\bkappabar{\bar{\bkappa}}
\def\bkappahat{\widehat{\bkappa}}
\def\bkappa{\boldsymbol{\kappa}}
\def\subsmbbeta{_{\scriptscriptstyle \bbeta}}
\def\subsmbbetan{_{\scriptscriptstyle \bbeta_0}}
\def\subismbbetan{_{i, {\scriptscriptstyle \bbeta_0}}}
\def\subismbbeta{_{i, {\scriptscriptstyle \bbeta}}}

Unlike $\balphahat$ and $\bgammahat$, estimation errors of $\widehat{h}(\cdot)$ and $\widehat{r}(\cdot)$ are larger in rate than the desirable parametric rate $n^{-1/2}$ since they are estimated using non-parametric approaches like kernel smoothing. In addition, removing the large non-parametric estimation biases from the biases of the resulting $\bbetahat\subatrel$ is particularly challenging due to the bias and variance trade-off in non-parametric regression. To motivate our strategy for mitigating such biases, we consider the estimation of $\c\trans\bbeta_0$, an arbitrary linear functional of $\bbeta_0$ where $\|\c\|_2=1$, and study the first order (over-fitting) bias incurred by $\widehat{h}(\cdot)$ and $\widehat{r}(\cdot)$ in $\c\trans\bbetahat\subatrel$. 
The essential bias terms of $n^{1/2}(\c\trans\bbetahat\subatrel-\c\trans\bbeta_0)$ arising from the non-parametric components can be asymptotically expressed as
\begin{equation}
\begin{split}
\Delta_1=&\frac{1}{\sqrt{n}}\sum_{i=1}^n\bar\omega(\X_i)\bkappa\subismbbetan\left\{Y_i-\mbar(\X_i)\right\}\{\widehat{h}(\Z_i)-\bar{h}(\Z_i)\};\\
\Delta_2=&\frac{1}{\sqrt{n}}\sum_{i=1}^n\bar{\omega}(\X_i)\bkappa\subismbbetan\breve g\{\bar m(\X_i)\}\{\rhat(\Z_i)-\bar r(\Z_i)\}\\
&-\frac{\sqrt{n}}{N}\sum_{i=n+1}^{N+n}\bkappa\subismbbetan\breve g\{\bar m(\X_i)\}\{\rhat(\Z_i)-\bar r(\Z_i)\},
\end{split}    
\label{equ:bias}
\end{equation}
where $\bkappa\subismbbeta=\c\trans\bJ\subsmbbeta^{-1}\A_i$
$\breve g(a)=\dot{g}\{g^{-1}(a)\}$, $\dot{g}(x)=dg(x)/dx > 0$, $\bJ\subsmbbeta=\rE_{0}\{\dot{g}(\A\trans\bbeta)\A\A\trans\}$ is the limit of $\bJhat\subsmbbeta=N^{-1}\sum_{i=n+1}^{n+N}\dot{g}(\A_i\trans\bbeta)\A_i\A_i\trans$,  $\bar\omega(\X)=\exp\{\bpsi\trans\bar\balpha+\bar h(\Z)\}$, $\bar m(\X)=g\{\bphi\trans\bar{\bgamma}+\bar r(\Z)\}$, $\hbar(\Z)$, $\rbar(\Z)$, $\balphabar$, $\bgammabar$, and  $\bar\bbeta$ are the respective limits of $\hhat(\Z)$, $\rhat(\Z)$, $\balphahat$, $\bgammahat$ and $\bbetahat\subatrel$. These limiting values are not necessarily true model parameter values due to potential model mis-specification.

When $m(\X)$ and $\omega(\X)$ are specified fully nonparametrically as those in \cite{rothe2015semiparametric} and \cite{chernozhukov2016double}, a standard cross-fitting strategy can removing terms like $\Delta_1$ and $\Delta_2$ by leveraging $\bar m(\X)=\mu(\X)$ and $\bar \omega(\X)=\bbw(\X)$ and utilizing the orthogonality between the ``residual" of $S$ or $Y$ on the covariates $\X$ and the functional space of $\X$. However, simply adopting cross-fitting is not sufficient for the current setting  because such orthogonality does not hold due to the potential mis-specifications of $m(\cdot)$ and $\omega(\cdot)$ in (\ref{model}). To overcome this challenge, we impose moment condition constraints on the nonparametric components $\bar r(\Z)$ and $\bar h(\Z)$ in that: for any measurable function $f(\cdot)$ of the covariates $\Z$,
\begin{align}
&\rE_1\left[ \bbw(\X)\bkappa\subsmbbetan\left(Y-g\left\{\bPhi\trans\bar{\bgamma}+\bar{r}(\Z)\right\}\right)f(\Z)\right]=0;\label{equ:mom:1}\\
&\rE_1\left[\exp\{\bpsi\trans\bar{\balpha}+\bar{h}(\Z)\}\bkappa\subsmbbetan\breve{g}\{\mu(\X)\}f(\Z)\right]=\rE_{0}\left[\bkappa\subsmbbetan\breve{g}\{\mu(\X)\}f(\Z)\right].
\label{equ:mom:2}
\end{align}
\begin{remark}
{When the density ratio model is correct, moment condition (\ref{equ:mom:2}) is naturally satisfied and solving (\ref{equ:mom:2}) for $\bar h(\cdot)$ leads to the true $h_0(\cdot)$. Constructing $\rbar(\cdot)$ under the moment condition (\ref{equ:mom:1}) will enable us to remove excess bias arising from the empirical error in estimating $\bar h(\cdot)$. On the other hand, when the imputation model $m(\X)$ is correct, condition (\ref{equ:mom:1}) holds and solving (\ref{equ:mom:1}) for $\bar r(\cdot)$ leads to $r_0(\cdot)$. And similarly, constructing $\hbar(\cdot)$ under (\ref{equ:mom:2}) will enable us to remove bias from the error in estimating $\bar r(\cdot)$. See our theoretical analyses given in Section~\ref{sec:thm} and Appendix~\ref{sec:proof} for more details on these points.
}
\end{remark}

\subsection{Estimation Procedure for $\bbetahat\subatrel$}\label{sec:method:spec}

We next detail estimation procedures for $\bbetahat\subatrel$ under the constraints of the moment conditions (\ref{equ:mom:1}) and (\ref{equ:mom:2}). Here we mainly focus on classic local regression approaches for low dimensional and smooth nonparametric components $r(\cdot)$ and $h(\cdot)$. In Appendix~\ref{sec:app:ml}, we propose a more general construction procedure that can learn $r(\cdot)$ and $h(\cdot)$ using arbitrary modern machine learning algorithms (e.g. random forest and neural network).  Similar to \cite{chernozhukov2016double}, we adopt cross-fitting on the source sample to eliminate the dependence between the estimators and the samples on which they are evaluated, and remove the first order bias $\Delta_1$ and $\Delta_2$ through concentration. Specifically, we randomly split the source samples into $K$ equal sized disjoint sets, indexed by $\Isc_1,\ldots,\Isc_K$, with $\{1,...,n\}=\cup_{k=1}^K \Isc_k$ and denote 
$\Isc\subminusk = \{1,..,n\}\setminus\Isc_k$. 

Equations (\ref{equ:mom:1}) and (\ref{equ:mom:2}) involve not only $r(\cdot)$ and $h(\cdot)$ but also other unknown parameters that needed to be estimated. To this end, first obtain preliminary estimators for $\omega(\X)$ and $m(\X)$ via standard semiparametric regression as $\widetilde\omega\supfk(\X)=\exp\{\bpsi\trans\widetilde\balpha\supfk+\widetilde h\supfk(\Z)\}$ and $\widetilde m\supfk(\X)=g\{\bphi\trans\widetilde\bgamma\supfk+\widetilde r\supfk(\Z)\}$ on $\Isc\subminusk\cup\{n+1,\ldots,n+N\}$, where the nonparametric components can be estimated with either sieve \citep{beder1987sieve} or profile kernel/backfitting \citep{lin2006semiparametric}. {Here, we take sieve as an example. Let $\b(\Z)$ be some basis function of $\Z$ with growing dimension, e.g. Hermite polynomials as specified by Assumption \ref{asu:a3} in Appendix \ref{sec:app:kern}. Denote by $\bPsi = (\bpsi\trans,\b(\Z)\trans)\trans$ and $\bPhi = (\bphi\trans,\b(\Z)\trans)\trans$. We solve
\begin{alignat}{2}
&\frac{K}{n(K-1)}\sum_{i\in\Isc\subminusk}\bPsi_i\exp(\btheta_w\trans\bPsi_i)+\lambda_1 (0,\btheta_{w,\text{-}1}\trans)\trans=\frac{1}{N}\sum_{i=n+1}^{n+N}\bPsi_i; &\quad& \mbox{with } \btheta_w = (\balpha\trans,\bfeta\trans)\trans \label{equ:pre:nui:1}\\
&\frac{K}{n(K-1)}\sum_{i\in\Isc_{\text{-}k}}\bPhi_i\left\{Y_i-g(\btheta_m\trans\bPhi_i)\right\}+\lambda_2 (0,\btheta_{m,\text{-}1}\trans)\trans=\bzero, &\quad& \mbox{with } \btheta_m = (\bgamma\trans,\bxi\trans)\trans
\label{equ:pre:nui:2}
\end{alignat}
to obtain the estimators $\widetilde\btheta_w\supfk = (\widetilde\balpha\supfktrans,\widetilde\bfeta\supfktrans)\trans$, $\widetilde\btheta_m\supfk =(\widetilde\bgamma\supfktrans, \widetilde\bxi\supfktrans)\trans$ for $\btheta_w$ and $\btheta_m$, and $\widetilde h\supfk(\Z)=\b\trans(\Z)\widetilde\bfeta\supfk$, $\widetilde r\supfk(\Z)=\b\trans(\Z)\widetilde\bxi\supfk$. Here we include ridge penalties to improve the training stability, with the two tuning parameters $\lambda_1,\lambda_2=o_{p}(n^{-1/2})$. Suppose that $\widetilde\omega\supfk(\X)$ and $\widetilde m\supfk(\X)$ approach some limiting models denoted as $\omega^*(\X)=\exp\{\bpsi\trans\balpha^*+h^*(\Z)\}$ and $m^*(\X)=g\{\bphi\trans\bgamma^*+r^*(\Z)\}$. Certainly, we have that $\omega^*(\X)=\bbw(\X)$ when the density ratio model is correctly specified, and $m^*(\X)=\mu(\X)$ when imputation model is correct. Then we solve the estimating equation for $\bbeta$:
\[
\frac{K}{n(K-1)}\sum_{i\in\Isc\subminusk}\widetilde{\omega}\supfk(\X_i)\A_i\{Y_i-\widetilde{m}\supfk(\X_i)\}+\frac{1}{N}\sum_{i=n+1}^{N+n}\A_i\{\widetilde{m}\supfk(\X_i)-g(\A_i\trans\bbeta)\}=\mathbf{0},
\]
Denote its solution as $\widetilde\bbeta\supfk$, a preliminary estimator consistent for $\bbeta_0$ when at least one nuisance model is correct but typically not achieving the desirable parametric rate as our final goal.

One might improve the convergence rate of the remainder bias of $\widetilde\balpha\supfk$ and $\widetilde\bgamma\supfk$ by further using cross-fitting on the nonparametric components in estimating equations (\ref{equ:pre:nui:1}) and (\ref{equ:pre:nui:2}); see \cite{newey2018cross}. While the so called ``plug-in" or simultaneous M-estimation $\widetilde\balpha\supfk$ and $\widetilde\bgamma\supfk$ can be shown to be $n^{1/2}$-consistent and asymptotically normal under certain smoothness and regularity conditions \citep{shen1997methods,chen2007large}, and thus satisfy our requirement (see Assumption \ref{asu:3} and Proposition \ref{prop:1}). Therefore, one could simply set $\widehat\balpha\supfk=\widetilde\balpha\supfk$ and $\widehat\bgamma\supfk=\widetilde\bgamma\supfk$ as the estimator of the parametric components in the final nuisance models. Consequently, their limiting (true) values are also identical: $\bar\balpha=\balpha^*$ and $\bar\bgamma=\bgamma^*$. In the following part of this section, we choose this construction.

\begin{remark}
Equations (\ref{equ:pre:nui:1}) and (\ref{equ:pre:nui:2}) are not the only choices for specifying $\balpha$ and $\bgamma$. In our framework, $\balpha$ and $\bgamma$ could be estimated through any estimating equations ensuring their $n^{1/2}$-consistency for some limiting parameters equal to the true ones when the corresponding nuisance models are correct. This flexibility is particularly useful when the intrinsic efficiency \citep{tan2010bounded,rotnitzky2012improved} of our estimator is further desirable, i.e. $\c\trans\bbetahat\subatrel$ is the most efficient among all the doubly robust estimators when $\omega(\cdot)$ is correct and $m(\cdot)$ has some wrong specification. Interestingly, we find that one could elaborate an estimating procedure for $\bgamma$ to realize this property and shall leave relevant details in Appendix \ref{sec:app:disc:detail:intr}. 
\label{rem:par}
\end{remark}


}

\def\subsmbbeta{_{\scriptscriptstyle \bbeta}}
\def\subsmbbetan{_{\scriptscriptstyle \bbeta_0}}
\def\subismbbetan{_{i, {\scriptscriptstyle \bbeta_0}}}
\def\subismbbetahatfk{_{i, {\scriptscriptstyle \bbetahat\supfk}}}
\def\bzero{\mathbf{0}}

Then we construct the calibrated estimating equations for the nonparametric nuisance components based on $\widehat\balpha\supfk$, $\widehat\bgamma\supfk$ and the preliminary estimators. Let $K(\cdot)$ represent some kernel function satisfying $\int_{\mathbb{R}^{p_{\z}}}K(\z)d\z=1$ and define that $K_h(\z)=K(\z/h)$. Localizing the terms in (\ref{equ:mom:1}) and (\ref{equ:mom:2}) with $K_h(\cdot)$, we solve for $r(\z)$ and $h(\z)$ respectively from
\begin{equation}
\begin{split}
\frac{1}{|\Isc\subminusk|}\sum_{i\in \Isc\subminusk} & K_h(\Z_i-\z)\bkappahat\subismbbetahatfk\widetilde{\omega}\supfk(\X_i)\left[Y_i-g\left\{\bphi_i\trans\widehat{\bgamma}\supfk+r(\z)\right\}\right]=\bzero;\\
\frac{1}{|\Isc\subminusk|}\sum_{i\in\Isc\subminusk}&K_h(\Z_i-\z)\bkappahat\subismbbetahatfk\breve{g}\{\widetilde m\supfk(\X_i)\} \exp\left\{\bpsi_i\trans\widehat{\balpha}\supfk+h(\z)\right\}\\
=\frac{1}{N}\sum_{i=n+1}^{n+N} &K_h(\Z_i-\z)\bkappahat\subismbbetahatfk\breve{g}\{\widetilde m\supfk(\X_i)\}.    
\end{split}
\label{equ:kern:mom}    
\end{equation}
where $\bkappahat\subismbbeta=\c\trans\bJhat_{\bbeta}^{-1}\A_i$.
Equations in (\ref{equ:kern:mom}) calibrate the nonparametric components to ensure the orthogonality between their score functions and the functional space of $\Z$, which is necessary for removing the bias terms introduced in (\ref{equ:bias}). In contrasts, the parametric component could include different sets of covariates from $\Z$, and there is no need to calibrate them. This substantially distinguishes our framework from existing methods \citep[e.g.]{smucler2019unifying,tan2020model} utilizing a similar calibration idea to handle high dimensional sparse nuisance models . 

\def\Ihat{\widehat{I}}

\begin{remark}
\label{rem:sign:split}
If the weights $\bkappahat\subismbbetahatfk=\c\trans\bJhat_{\widetilde{\bbeta}\supfk}^{-1}\A_i$ have the same sign for a majority of the subjects $i\in\Isc\subminusk\cup\{n+1,\ldots,n+N\}$, both equations in (\ref{equ:kern:mom}) have an unique solution for each $\z$, denoted as $\rhat\supfk(\Z)$ and $\widehat h\supfk(\Z)$. In practice, it is more likely that $\bkappahat\subismbbetahatfk$ can be positive for some subjects and negative for others, in which case (\ref{equ:kern:mom}) can be irregular and ill-posed, leading to inefficient estimation. One simple strategy to overcome this is to expand the nuisance imputation models to allow $h$ and $r$ to differ among those with $\bkappahat\subismbbetahatfk\ge 0$ versus those with $\bkappahat\subismbbetahatfk$. Specifically, we may solve for 
\begin{equation}
\begin{split}
\frac{1}{|\Isc\subminusk|}\sum_{i\in \Isc\subminusk} \left[ \begin{matrix}\Ihat_{+, i}\supfk \\ \Ihat_{-, i}\supfk \end{matrix}\right]
& K_h(\Z_i-\z)\bkappahat\subismbbetahatfk\widetilde{\omega}\supfk(\X_i)\left[Y_i-g\left\{\bphi_i\trans\widehat{\bgamma}\supfk+\Ihat_{+, i}\supfk r_+(\z) + \Ihat_{-, i}\supfk r_-(\z) \right\}\right]=\bzero;\\
\frac{1}{|\Isc\subminusk|}\sum_{i\in\Isc\subminusk}\left[ \begin{matrix}\Ihat_{+, i}\supfk \\ \Ihat_{-, i}\supfk \end{matrix}\right] &K_h(\Z_i-\z)\bkappahat\subismbbetahatfk\breve{g}\{\widetilde m\supfk(\X_i)\} \exp\left\{\bpsi_i\trans\widehat{\balpha}\supfk+\Ihat_{+, i}\supfk h_+(\z) + \Ihat_{-, i}\supfk h_-(\z) \right\}\\
=\frac{1}{N}\sum_{i=n+1}^{n+N}\left[ \begin{matrix}\Ihat_{+, i}\supfk \\ \Ihat_{-, i}\supfk \end{matrix}\right] & K_h(\Z_i-\z)\bkappahat\subismbbetahatfk\breve{g}\{\widetilde m\supfk(\X_i)\},
\end{split}
\label{equ:kern:mom:sign}    
\end{equation}
where $\Ihat_{+, i}\supfk = I(\bkappahat\subismbbetahatfk\ge 0)$ and $\Ihat_{-, i}\supfk = I(\bkappahat\subismbbetahatfk< 0)$. Then we take $\widehat m\supfk(\X_i)=g\{\bphi_i\trans\widehat{\bgamma}\supfk+\Ihat_{+, i}\supfk r_+(\Z_i) + \Ihat_{-, i}\supfk r_-(\Z_i) \}$ and $\widehat \omega\supfk(\X_i)=\exp\left\{\bpsi_i\trans\widehat{\balpha}\supfk+\Ihat_{+, i}\supfk h_+(\Z_i) + \Ihat_{-, i}\supfk h_-(\Z_i) \right\}$. 
With this modification, our construction still effectively removes $\Delta_1$ and $\Delta_2$ as one could trivially analyze the two disjoint set of samples separately, and combine their convergence rates at last. 
\end{remark}

After obtaining $\widehat r\supfk(\cdot)$ and $\widehat h\supfk(\cdot)$ for each $k\in\{1,2,\ldots,K\}$, we take $\widehat\omega\supfk(\X_i)=\exp\{\bpsi_i\trans\widehat\balpha\supfk+\widehat h\supfk(\Z_i)\}$, $\widehat m\supfk(\X_i)=g\{\bphi_i\trans\widehat\bgamma\supfk+\widehat r\supfk(\Z_i)\}$, $\widehat m(\X_i)=K^{-1}\sum_{k=1}^K\widehat m\supfk(\X_i)$, and plug them into the cross-fitted version of the estimating equation (\ref{equ:dr}) written as:
\begin{equation}
\frac{1}{n}\sum_{k=1}^K\sum_{i\in\Isc_k}\omegahat\supfk(\X_i)\A_i\left\{Y_i-\mhat\supfk(\X_i)\right\}+\frac{1}{N}\sum_{i=n+1}^{N+n}\A_i\{\mhat(\X_i)-g(\A_i\trans\bbeta)\}=\mathbf{0}.
\label{equ:dr:cross}
\end{equation}
Let the solution of (\ref{equ:dr:cross}) be $\bbetahat\subatrel$ and we take $\c\trans\bbetahat\subatrel$ as the estimation for $\c\trans\bbeta_0$. For interval estimation of $\c\trans\bbeta_0$, we use bootstrap, which appears to have better numerical performance than using the asymptotic variance estimated directly by the moment estimator.



\section{Theoretical analysis}\label{sec:thm}

Assume that $\rho=n/N=O(1)$, $K=O(1)$. For any vector $\a$, let $\|\a\|_2$ represent its $\ell_2$-norm. 
Let $\mathcal{Z}$ and $\mathcal{X}$ represent the domains of $\Z$ and $\X$ respectively. Assume that dimensionality of $\A$, $p_{\bphi}$ and $p_{\bpsi}$ are fixed. We then introduce three sets of assumptions as follows.

\begin{assumption}[Regularity conditions]
\label{asu:1}
There exists a constant $C_L>0$ such that $|\dot g(a)-\dot g(b)|\leq C_L|a-b|$ for any $a,b\in\mathbb{R}$. $\bbeta_0$ belongs to a compact space. $\A_i$ belong to a compact set 
and has a continuous differential density on both populations $\Ssc$ and $\Tsc$. There exists a constant $C_U>0$ such that $\rE_j|Y|^2+\rE_1\bar\omega^4(\X)+\rE_j\breve g^4\{\bar m(\X)\}+\rE_j\|\bphi\|_2^4+\rE_j\|\bpsi\|_2^8<C_U$, for $j\in\{0,1\}$. The information matrix $\bJ_{\bbeta_0}$ has its all eigenvalues bounded away from $0$ and $\infty$.
\end{assumption}

\begin{assumption}[Specification of the nuisance models]
\label{asu:2}
At least one of the following two conditions holds: (i) $\bbw(\X)=\exp\{\bpsi\trans\balpha_0+h_0(\Z)\}$ for some $\balpha_0$ and $h_0(\cdot)$; or (ii) $\mu(\X)=g\{\bphi\trans\bgamma_0+r_0(\Z)\}$ for some $\bgamma_0$ and $r_0(\cdot)$. 
\end{assumption}

\begin{assumption}[Estimation error of the nuisance models]
\label{asu:3}
The nuisance estimators satisfy that (i) $n^{1/2}(\widehat\balpha\supfk-\bar\balpha)$ and $n^{1/2}(\widehat\bgamma\supfk-\bar\bgamma)$ is asymptotically normal with mean $\bzero$ and finite variance; (ii) for every $k\in\{1,2,\ldots,K\}$ and $j\in\{0,1\}$: 
\begin{align*}
&\rE_1\{\widehat h\supfk(\Z)-\bar h(\Z)\}^2+\rE_j\{\widehat r\supfk(\Z)-\bar r(\Z)\}^2=o_p(n^{-1/2});\\
&\sup_{\z\in\mathcal{Z}}|\widehat h\supfk(\z)-\bar h(\z)|+|\widehat r\supfk(\z)-\bar r(\z)|=o_p(1).
\end{align*}
\end{assumption}



\begin{remark}
Assumption \ref{asu:1} is reasonable and commonly used for asymptotic analysis of $M$-estimation such as logistic regression \citep{van2000asymptotic}. Assumption on the compactness of the domain of $\A_i$ could be relaxed to accommodate unbounded covariates with regular tail behaviours. Assumption \ref{asu:2} assumes that at least one nuisance model is correctly specified, and the nonparametric component in the possibly wrong model satisfies the moment constraints (\ref{equ:mom:1}) or (\ref{equ:mom:2}). Similar to the classic double robustness condition for the parametric nuisance models \citep{bang2005doubly,qin2008efficient}, the parametric part from the wrong model in our method could be arbitrarily specified. 
\label{rem:1}
\end{remark}

Assumption \ref{asu:3}(ii) assumes that both the nonparametric components have their mean squared errors (MSE) below $o_p(n^{-1/2})$, known as the rate doubly robust assumption \citep{smucler2019unifying}. With a similar spirit to \cite{chernozhukov2016double}, our Assumption \ref{asu:3} is imposed directly on the calibrated estimators $\widehat h\supfk(\cdot)$ and $\widehat r\supfk(\cdot)$ regardless of their specific estimation procedures, to preserve the generality. Justification of Assumption \ref{asu:3} for the nuisance estimators obtained through smooth regression introduced in Section \ref{sec:method:spec} is not standard because the estimating equations in (\ref{equ:kern:mom}) involve the nuisance preliminary estimators impacting the calibrated estimator through their empirical errors. We present this result as Proposition \ref{prop:1} and its proof in Appendix~\ref{sec:app:kern}, leveraging existing literature about sieve and kernel approaches \citep{fan1995local,carroll1998local,shen1997methods,chen2007large}.

\begin{proposition}
Under Assumption \ref{asu:1} and Assumptions \ref{asu:a1}--\ref{asu:a3} presented in Appendix \ref{sec:app:kern} about regularity, smoothness and specification of the sieve and kernel functions, Assumption \ref{asu:3} holds for our mainly proposed nuisance estimators in Section \ref{sec:method:spec}.
\label{prop:1}
\end{proposition}




Different from the sieve and kernel approaches introduced in Section \ref{sec:method:spec}, when there is high dimensional $\Z$ and the nonparametric components are estimated using modern machine learning approaches like lasso and random forest, our debiased method introduced in Appendix \ref{sec:app:disc} is used to construct the parametric nuisance components. We demonstrate in Appendix \ref{sec:app:disc} that such debiased estimation will satisfy Assumptions \ref{asu:3}(i) when the machine learning estimators for the nonparametric components have good quality.




Now we present the main theoretical results about the consistency and asymptotic validity of our estimator $\c\trans\bbetahat\subatrel$ in Theorem \ref{thm:1} with its proof found in Appendix~\ref{sec:proof}.
\begin{theorem}
Under Assumptions \ref{asu:1} to \ref{asu:3}, it holds that $\|\bbetahat\subatrel-\bbeta_0\|_2=o_p(1)$ and 
\[
\sqrt{n}(\c\trans\bbetahat\subatrel-\c\trans\bbeta_0)=\frac{1}{\sqrt{n}}\sum_{i=1}^nF_i^{\Ssc}+\frac{\sqrt{n}}{N}\sum_{n+1}^{n+N}F_i^{\Tsc}+\sqrt{n}\bzeta_{\alpha}\trans(\widehat\balpha-\bar\balpha)+\sqrt{n}\bzeta_{\gamma}\trans(\widehat\bgamma-\bar\bgamma)+o_p(1),
\]
where $F_i^{\Ssc}=\bar\omega(\X_i)\A_i\left\{Y_i-\bar m(\X_i)\right\}$, $F_i^{\Tsc}=\A_i\{\bar m(\X_i)-g(\A_i\trans\bbeta)\}$, 
\begin{align*}
\bzeta_{\alpha}&=\rE_1\bar\omega(\X)\bkappa\subsmbbetan\left[Y-g\{\bphi\trans\bar{\bgamma}+\bar r(\Z)\}\right]\bpsi,\\    
\bzeta_{\gamma}&=\rE_1\bar{\omega}(\X)\bkappa\subsmbbetan\breve g\{\bar m(\X)\}\bphi-\rE_0\bkappa\subsmbbetan\breve g\{\bar m(\X)\}\bphi,
\end{align*}
$\widehat\balpha=K^{-1}\sum_{k=1}^K\widehat\balpha\supfk$, and $\widehat\bgamma=K^{-1}\sum_{k=1}^K\widehat\bgamma\supfk$. Consequently, $n^{1/2}(\c\trans\bbetahat\subatrel-\c\trans\bbeta_0)$ weakly converges to Gaussian distribution with mean $\bzero$ and variance of order $1$.
\label{thm:1}
\end{theorem}


\begin{remark}
When Assumption \ref{asu:2}(i) holds, i.e. the density ratio is correctly specified, one have that $\bzeta_{\gamma}=\bzero$ so $\widehat\bgamma\supfk-\bar\bgamma$ has no impact on the asymptotic expansion $\c\trans\bbetahat\subatrel$. Similarly, when the imputation model is correct, $\bzeta_{\alpha}=\bzero$ and $\widehat\balpha\supfk-\bar\balpha$ has no impact on $\c\trans\bbetahat\subatrel$. When both nuisance models are correctly specified, $\c\trans\bbetahat\subatrel$ is a semiparametric efficient estimator for $\c\trans\bbeta_0$ in our case of covariate shift regression \citep{hahn1998role}.
\end{remark}


\def\ba{\mathbf{a}}
\def\bb{\mathbf{b}}
\def\bc{\mathbf{c}}
\def\bd{\mathbf{d}}
\def\subBE{_{\scriptscriptstyle \sf BE}}
\def\subKM{_{\scriptscriptstyle \sf KM}}

\section{Simulation studies}\label{sec:simu}

We conduct simulation studies to investigate the performance of the ATReL method and compare it with existing doubly robust approaches. We consider four different data generating mechanisms concerning specification of the nuisance models. Throughout, we let $n=500$ and $N = 1000$. To generate the data, we first generate $\V=(V_1,V_2, ..., V_7)\trans$ from $\Nsc(\bzero,\Sigma_V)$ where $\Sigma_V=(\sigma_{ij})_{7\times 7}$, $\sigma_{ij}=1$ when $i=j$, $\sigma_{ij}=0.3$ when $(i,j)\mbox{ or }(j,i)\in\{(1,2),(1,3),(3,4),(3,5)\}$, $\sigma_{ij}=0.15$ when $(i,j)\mbox{ or }(j,i)\in\{(1,6),(1,7),(5,6),(5,7)\}$, and $\sigma_{ij}=0$ otherwise. Then we obtain each $\Xtilde_j$ by truncating $V_j$ with $(-1.5,1.5)$ and standardizing it, and take
\[
\W=\left\{1, \exp(0.5\Xtilde_1), \frac{\Xtilde_2}{1+\exp(\Xtilde_3)} ,\left(\frac{\Xtilde_1\Xtilde_3}{5}+0.6\right)^3, \Xtilde_4,...,\Xtilde_7\right\}\trans
\]
as a nonlinear transformation of $\bXtilde=(1,\Xtilde_1,\Xtilde_2,\ldots,\Xtilde_7)\trans$. Based on this, we consider four configurations for the underlying data generating mechanisms introduced below as the configurations indexed by (i)--(iv). First, we set $Z=\Xtilde_1$ and generate the source indication $S$ given $\bXtilde$ by ${\rm P}(S = 1 \mid \bXtilde) = g\{\ba_w\trans\W + \ba_x\trans\bXtilde+ h_x(Z)\}$ where

\begin{enumerate}
    \item[(i)] $\ba_w = (-1,0,-0.4,-0.4,-0.15,-0.15,0,0)\trans$, $\ba_x = \bzero$, and
    $h_x(Z)=0.6Z^2\cdot I(\lvert Z \rvert <1.5)+\{ 0.6(\lvert Z \rvert-1.5)+1.35\}\cdot I(\lvert Z \rvert \ge 1.5)$.
    \item[(ii)] The same as Configurations (i).
    \item[(iii)] $\ba_w = \bzero$, $\ba_x = (0,-0.2,-0.4,-0.4,-0.2,-0.2,0,0)\trans$, and
    $h_x(Z)=0.5{\lvert Z \rvert}^3\cdot I(\lvert Z \rvert <1.5)+\{0.5\cdot1.5^3+(\lvert Z \rvert-1.5)\}\cdot I(\lvert Z \rvert \ge 1.5)$.
    \item[(iv)] $\ba_w = \bzero$, $\ba_x = (0,-0.4,-0.4,-0.4,-0.15,-0.15,0,0)\trans$, and $h_x(Z)=0$.
\end{enumerate}
In Configurations 1 and 2, set the observed covariates as $\X=(1,X_1,X_2,\ldots,X_7)\trans$ where
\[
\widetilde{X}_2 =  0.8X_2  -0.2sin(\frac{3}{4}\pi Z)\cdot I(S=0);\quad\widetilde{X}_3 =  0.8X_3 -0.2sin(\frac{3}{4}\pi Z)\cdot I(S=0),
\]
and $X_j=\Xtilde_j$ for all $j\neq 2,3$. While in Configurations 3 and 4, we simply set $\X=\bXtilde$. Then we generate $Y$ given $\X$ by ${\rm P}(Y = 1 \mid \X) = g\{\bb_w\trans\W + \bb_x\trans\X+ r_x(Z)\}$, where 
\begin{enumerate}
    \item[(i)]  $\bb_w =\bzero$, $\bb_x = (0,0.5,0.5,0.5,0.3,0.3,0.15,0.15)\trans$, $r_x(Z)=-0.4 \cdot sin(\frac{3}{4}\pi Z)$.\\
    \item[(ii)] $\bb_w =\bzero$, $\bb_x = (0,0.5,0.5,0.5,0.3,0.3,0.15,0.15)\trans$, $r_x(Z)= 0$.\\
    \item[(iii)] $\bb_w = (-0.5,0.5,0.8,0.3,-0.3,-0.2,0.15,0.15)\trans$, $\bb_x =\bzero$, $r_x(Z)=-0.6 \cdot sin(\frac{3}{4}\pi Z).$\\
    \item[(iv)] $\bb_w = (-0.8,0.5,0.5,0.5,0.3,0.3,0.15,0.15)\trans$, $\bb_x =\bzero$, $r_x(Z)=-0.4 \cdot sin(\frac{3}{4}\pi Z).$\\
\end{enumerate}
In all the four configurations, we set $\A = (1, X_1, ..., X_3)\trans$. For each generated dataset, we fit the following nuisance models to estimate $\bbeta_0$: 
\begin{enumerate}
\item[(a)] Parametric nuisance models (Parametric): the importance weight model is chosen as the logistic model of $S$ against $\bPsi=\X$ and the imputation model is specified as the logistic model of $Y$ against $\bPhi=\X$.

\item[(b)] Semi-non-parametric nuisance models (ATReL): ${\rm P}(S=1 \mid \X)=g\{\bPsi\trans\balpha+h(Z)\}$ and  ${\rm P}(Y=1\mid\X)=g\{\bPhi\trans\bgamma+r(Z)\}$, where $\bPsi=\X$, $\bPhi=\X$, and $Z=X_1$.

\item[(c)] Double machine learning with flexible basis expansions (DML$\subBE$): the nuisance models regress $Y$ or $S$ on features combining together $\X$, natural splines of each $X_j$ with order $4$ and all the interaction terms of these natural splines. Due to high dimensionality of the bases, we use a combination of $\ell_1$ and $\ell_2$ penalties for regularization.

\item[(d)] Double machine learning with kernel machine (DML$\subKM$): both models are estimated using support vector machine with the radial basis function kernel.

\end{enumerate}

Our data generation and model specification have a similar spirit as \cite{kang2007demystifying} and \cite{tan2020model}. In Configurations (i) and (ii), our semi-non-parametric imputation model correctly characterizes $Y\mid\X$ while our importance weight model is mis-specified. Parametric approach (a) has its imputation model correctly specified under Configuration (ii) but misses the nonlinear function $r(Z)$ under (i). Also note that under (ii), nonparametric component included in the imputation model of our method is redundant for the logistic linear model of ${\rm P}(Y=1\mid\X)$. Similar logic applies to Configurations (iii) and (iv) with the status of the imputation model and importance weight model interchanged. More implementing details of (a)--(d) are presented in Appendix~\ref{sec:app:results:sim}.


Performance of the four approaches are evaluated through root mean square error, bias and coverage probability of the 95\% confidence interval in terms of estimating and inferring $\beta_0,\beta_1,\beta_2,\beta_3$, as summarized in Tables \ref{tab:simu:IMPcornp}--\ref{tab:simu:IWcornonp} of Appendix~\ref{sec:app:results:sim} for configurations (i)--(iv) respectively. The mean square error and absolute bias averaged over the target parameters, and the maximum deviance of the coverage probability from the nominal level $0.95$ among all parameters are summarized in Table~\ref{tab:simu:1}.

\begin{table}[htbp] 
\centering 
\caption{\label{tab:simu:1} Average root mean square error (RMSE), average absolute bias ($|$Bias$|$), and maximum deviance of coverage probability (CP) of the constructed CI from its nominal level $0.95$ over all parameters of the doubly robust estimators with different modeling strategies for the nuisance models: Parametric, ATReL, DML$\subBE$ and DML$\subKM$ under Configurations (i)--(iv), as introduced in Section \ref{sec:simu}.} 
\begin{tabular}{cccccc} 
\\[-1.8ex]\hline 
\hline \\[-1.8ex] 
Configurations & & Parametric  & ATReL & DML$\subBE$ & DML$\subKM$ \\
\hline \\[-1.8ex] 
(i) & Average RMSE & 0.141 & 0.123 & 0.179 & 0.153  \\ 
& Average $|$Bias$|$ & 0.065 & 0.030 & 0.108 & 0.058  \\ 
& Deviance of CP  & 0.04 & 0.02 & 0.11 & 0.10 \\
\hline\\[-1.8ex] 
(ii) & Average RMSE & 0.117 & 0.123 & 0.186 & 0.148   \\ 
& Average $|$Bias$|$ & 0.005 & 0.016 & 0.114 & 0.061 \\ 
& Deviance of CP  & 0.04 & 0.02 & 0.13 & 0.05 \\
\hline\\[-1.8ex] 
(iii) & Average RMSE & 0.207 & 0.134 & 0.142 & 0.144   \\ 
& Average $|$Bias$|$ & 0.092 & 0.019 & 0.036 & 0.062  \\ 
& Deviance of CP  & 0.13 & 0.02 & 0.02 & 0.09 \\
\hline\\[-1.8ex] 
(vi) & Average RMSE & 0.131 & 0.122 & 0.145 & 0.128   \\ 
& Average $|$Bias$|$ & 0.005 & 0.009 & 0.058 & 0.044  \\ 
& Deviance of CP  & 0.01 & 0.02 & 0.22 & 0.09 \\
\hline
\hline \\[-1.8ex] 
\end{tabular} 
\end{table} 

Under all configurations, ATReL achieves better performance, especially at least $48\%$ smaller average bias, than the two double machine learning approaches. Also, ATReL performs well in interval estimation with coverage probabilities on all parameters under all configurations falling in $\pm 0.02$ of the nominal level. In comparison, the Parametric method fails obviously on interval estimation of $\beta_1$ under (iii) because in the importance weighting model, nonparametric component is placed on the corresponding predictor. The two double machine learning approaches fail apparently on interval estimation of certain parameters, for example, Additive approach fails on interval estimation of $\beta_0$ under Configuration (i), (ii) and (iv) and Kernel machine fails on $\beta_1$ under Configuration (i), (iii) and (iv). These demonstrate that our method achieves better balance on the model complexity than the fully nonparametric/machine learning constructions, leading to consistently better performance on point and interval estimation. 

Our method has significantly smaller root mean square error than Parametric under (i) (relative efficiency being $0.89$) and (iii) (relative efficiency being $0.65$), with nonlinear effects in the nuisance models captured by our method and missed by the parametric approach. Under these two configurations, our method also has ($55\%$ under (i) and $79\%$ under (iii)) smaller average absolute bias than Parametric. While for (ii) and (iv) with the nonparametric components in our construction being redundant, performance of our method is close to the parametric approach. Thus, our nonparametric components modeling help to reduce bias and improve estimation efficiency in the presence of nonlinear effects while they basically do not hurt the efficiency when being redundant. 



\section{Transfer EHR phenotyping of rheumatoid arthritis across different time windows}\label{sec:real}

Growing availability of EHR data opens more opportunities for translational biomedical research \citep{kohane2012translational}. However, a major obstacle to realizing the full translational potential of EHR is the lack of precise definition of disease phenotypes needed for clinical studies. With a small number of gold standard labels for phenotypes, machine learning phenotyping algorithms based on both codified EHR features and clinical note mentions extracted using natural language processing (NLP) have been derived to improve the phenotype definition \cite{liao2019high}. For example, several phenotyping algorithms for rheumatoid arthritis (RA), a common autoimmune disease, have been developed and validated at multiple institutions in recent years \citep{liao2010electronic,carroll2012portability,yu2017surrogate}. Once the phenotyping algorithms become available, they are used to classify disease status for downstream tasks such as genomic association studies using EHR linked biobank data \citep{kohane2011using}.

Once a phenotyping algorithm is developed, it is often used repeatedly to classify disease status for patients in an EHR database which are often updated over time. For example, the RA algorithm developed by \cite{liao2010electronic} at Mass General Brigham (MGB) was trained in 2009 and validated again in 2020 \cite{huang2020impact}. Significant changes have occurred between 2009 and 2020: the EHR system at MGB was switched to EPIC and the International Classification of Diseases (ICD) system was changed from version 9 to version 10 around 2015 - 2016. Although the algorithm trained in \cite{liao2010electronic} appears to have stable performance for the 2020 data \cite{huang2020impact}, we investigated to what extent transfer learning can be used to automatically update the phenotyping algorithm over time. 
%
%
To this end, we considered training an RA EHR phenotyping algorithm to classify RA status for patients with EHR data from 2016 at MGB using training data from 2009. 

There are a total of 200 labeled patients with true RA status, $Y$, manually annotated via chart review. There are a total of $p=9$ demographic or EHR features, $\X$, available for training RA algorithm, including the total healthcare utilization ($X_1$),  NLP count of RA ($X_2$), NLP mention of tumor necrosis factor (TNF) inhibitor ($X_3$), NLP mention of bone erosion ($X_4$), age ($X_5$), gender ($X_6$), ICD count of RA ($X_7$), presence of TNF inhibitor prescription ($X_8$), and tested negative for rheumatoid factor ($X_9$), where we use $x\to \log(x+1)$ transformation for all count variables. Since NLP mentions of clinical terms are less sensitive to changes to the EHR coding system, we aim to develop an NLP feature only model for predicting $Y$ using $\A = (X_1, X_2, X_3, X_4)\trans$, for the EHR cohort of 2016 using labeled data from 2009 via transfer learning. Due to the co-linearity among $\A$, we convert $X_2$ into its orthogonal complement to $X_1$. For simplicity, we still denote the transformed covariates as $(X_1,X_2,X_3,X_4)\trans$. 


We implemented the doubly robust transfer learning approaches introduced in Section \ref{sec:simu}, including Parametric, ATReL, DML$\subBE$ and DML$\subKM$. Specific construction of the nuisance models in the four approaches are presented in Appendix~\ref{sec:app:results:real}. We also include the logistic model for $Y\sim\A$ simply fitted on the source data without adjusting for covariate shift, named as Source. For our proposed ATReL, we choose $Z$ as the NLP count of RA for non-parametric modeling since it is the most predictive feature in $\A$.

To evaluate the performance of the transfer learning, we additional performed chart review on $150$ subjects from the target population in 2016, denoted as $\Lsc_{16}$. We fit a logistic regression $Y \sim \A$ using these labeled observations in $\Lsc_{16}$ and denote the estimate for $\bbeta$ as $\bbetahat\subvalid$ to serve as gold standard benchmark. Fitted intercepts and coefficients of all methods are presented in Table~\ref{tab:real:beta2015} of Appendix~\ref{sec:app:results:real}. To evaluate the estimation performance of a derived estimator $\bbetahat$ according to our practical needs, we calculate the following metrics:
\begin{itemize}
    \item[] {\bf AUC}. Area under the receiver operating characteristic (ROC) curve evaluated with the labels. For the Target estimator $\bbetahat\subvalid$, we use repeated sample-splitting for evaluation.
    
    \item[] {\bf RMSPE}. Relative mean square prediction error to $\bbetahat\subvalid$ evaluated on the target data:
    \[
    \frac{\widehat\rE_{0}\{g(\A\trans\bbetahat\subvalid)-g(\A\trans\bbetahat)\}^2}{\widehat\rE_{0}\{g(\A\trans\bbetahat\subvalid)\}^2}.
    \]
    \item[] {\bf CC} with $\bbetahat\subvalid$. Classifier's correlation with that of $\bbetahat\subvalid$:
    \[
    \widehat{\rm Corr}_0\left\{I\left(g(\A\trans\bbetahat\subvalid)\geq \widehat\rE_{0}[g(\A\trans\bbetahat\subvalid)]\right),I\left(g(\A\trans\bbetahat)\geq \widehat\rE_{0}[g(\A\trans\bbetahat)]\right)\right\},
    \]

    \item[] {\bf FCR} v.s. $\bbetahat\subvalid$. False classification rate of $\bbetahat$'s classifier against that of $\bbetahat\subvalid$:
    \[
    \widehat{\rP}_0\left\{I\left(g(\A\trans\bbetahat\subvalid)\geq \widehat\rE_{0}[g(\A\trans\bbetahat\subvalid)]\right)\neq I\left(g(\A\trans\bbetahat)\geq \widehat\rE_{0}[g(\A\trans\bbetahat)]\right)\right\}.
    \]
\end{itemize}
Here $\widehat\rE_{0}$, $\widehat{\rP}_0$, and $\widehat{\rm Corr}_0(\cdot,\cdot)$  represent the empirical expectation, probability measure, and pearson correlation on the target population.
Evaluation results obtained with the target data and the validation labels are presented in Table~\ref{tab:realtable:1}. Our ATReL method attains the smallest estimation error among all the methods under comparison, with its relative efficiency of RMSPE being $0.21$ to the naive source estimator, $0.23$ to doubly robust estimator with parametric nuisance models, $0.17$ to double machine learning with flexible basis expansions, and $0.46$ to double machine learning with kernel machine. Also, among Source and all the transfer learning estimators, ATReL produces the largest AUC, as well as the closest classifiers to the gold standard target data estimator, i.e. attaining the largest CC with $\bbetahat\subvalid$ and smallest FCR v.s. $\bbetahat\subvalid$. Thus, by trading-off the parametric and nonparametric modeling strategies in a better way to adjust for the covariate shift, our method achieves better estimation performance than all existing methods.

\begin{table}[H]
\centering
\caption{Estimation performance of the source or transfer learning estimators evaluated with the validation labeled data and validation estimator denoted as Target. All included methods are as described in Sections \ref{sec:simu} and \ref{sec:real}. The evaluation metrics, as introduced in Section \ref{sec:real}, include AUC: area under the ROC curve; RMSPE: relative mean square prediction error; CC with $\bbetahat\subvalid$: classifier's correlation with that of $\bbetahat\subvalid$; FCR v.s. $\bbetahat\subvalid$: false classification rate against $\bbetahat\subvalid$.}
\label{tab:realtable:1}
\begin{tabular}{ccccccc}
\hline
\hline\\[-1.8ex] 
& Source & Parametric  &{ATReL} & DML$\subBE$ & DML$\subKM$ & {Target} \\
\hline\\[-1.8ex] 
AUC &   0.908 & 0.904 & 0.916 & 0.907 & 0.911 & 0.922 \\
RMSPE &  0.052 & 0.048 & 0.011 & 0.064 & 0.024 & 0 \\ 
\hline\\[-1.8ex] 
Prevalence  &  0.376 & 0.336 & 0.323 & 0.329 & 0.330 & 0.340 \\
\hline\\[-1.8ex] 
CC with $\bbetahat\subvalid$ &  0.890 & 0.880 & 0.970 & 0.910 & 0.930 & 1 \\ 
FCR v.s. $\bbetahat\subvalid$ & 0.050 & 0.060 & 0.010 & 0.050 & 0.030 & 0 \\  
\hline\hline\\[-1.8ex] 
\end{tabular}
\end{table}


\section{Discussion}\label{sec:disc}

\paragraph{Contribution and limitation.}
In this paper, we propose ATReL, a transfer regression learning approach using an imputation model to augment the importance weighting equation to achieve double robustness. Moreover, we propose a novel semi-non-parametric framework to construct the two nuisance models that achieves a better model complexity trade-off than existing doubly robust or double machine learning approaches. We show that $n^{1/2}$-consistency of our proposed estimator is guaranteed by a hybrid of the model double robustness of the parametric component and the rate double robustness of the nonparametric component. Simulation studies and the real example also demonstrate that our method is more robust and efficient than the existing fully parametric and double machine learning estimators. In our current approach, choice and specification of the nonparametric covariates $\Z$ really depend on one's prior knowledge or some preliminary analysis. Since it is crucial for us to properly choose the set of covariates in $\Z$ as well as its modeling strategy, it is desirable to further develop data-driven approaches to select the set and model of $\Z$ in our framework, to make ATReL more stable and usable in practice. We also notice some potential directions to generalize or enhance our current proposal and introduce them shortly as below with more details presented in Appendix~\ref{sec:app:disc}.

\paragraph{Sieve or modern machine learning estimation of the nonparametric parts.}

We propose some other choices in constructing the nuisance estimators alternative to the kernel smoothing method introduced in Section \ref{sec:method:spec}. Detailed construction procedures under these choices, including sieve and modern (black-box) machine learning algorithms are presented in Appendix~\ref{sec:app:disc}. First, we note that sieve can be naturally incorporated to solve the calibrated equations in (\ref{equ:kern:mom}) and achieve the same convergence properties as kernel. More importantly, we propose a construction procedure using arbitrary modern (nonparametric) machine learning algorithms to learn the nonparametric components in the nuisance models under our framework. This is substantially more challenging than the kernel or sieve constructions since we consider arbitrary black-box machine learning algorithms with no special forms, and thus it becomes more involving to derive nuisance estimators satisfying the moment conditions (\ref{equ:mom:1}) and (\ref{equ:mom:2}). To our best knowledge, similar problem has not been solved in existing literature. 



\paragraph{The $N\gg n$ scenario.}
In many application fields like EHR phenotyping studied in this paper, sample size of unlabeled data $N$ can usually be much larger than the size of labeled data $n$. Analysis of our method under such a $N\gg n$ scenario is of particular interests. 
It has been established that semi-supervised learning with $N\gg n$ unlabeled samples enables estimating varies types of target parameters more efficiently than the supervised method \citep[e.g.]{kawakita2013semi,azriel2016semi,gronsSSL2017,chakrabortty2018efficient,gronsbell2020efficient}. However, existing work is restricted to the setting where the unlabeled and labeled data are from the same population. In the presence of covariate shift, it is of interests to further investigate whether having $N\gg n$ (unlabeled) target samples would benefit our estimator. As we could tell, when the importance weight model is correct, similar results as \cite{kawakita2013semi} should apply in our case and the asymptotic variance of ATReL could be reduced compared with the estimator obtained under the $N\asymp n$ or $N<n$ scenarios. Study of this problem warrants future work.

\paragraph{Intrinsic efficient estimator.}
When the importance weight model is correctly specified while the imputation model may be wrong, asymptotic variance of our estimator is dependent of the parameters $\bar\gamma$ and $\bar r(\cdot)$. For purely fixed dimensional parametric nuisance models, there exists certain moment equations for the imputation parameters that grants one to get the most efficient doubly robust estimator among those with the same specification of the imputation model. This property is referred as intrinsic efficiency \citep{tan2010bounded,rotnitzky2012improved}. Under our semi-nonparemetric framework, flexibility on specifying the parametric parts of the nuisance models makes the intrinsic efficiency of our proposed estimator worthwhile considering. In Appendix~\ref{sec:app:disc:detail:intr}, we introduce a modified construction procedure for $\widehat m\supfk(\cdot)$ that calibrates its nonparametric part, and ensures the intrinsic efficiency of the estimator of $\c\trans\bbeta_0$, or more generally, any given smooth function of $\bbeta_0$.

\bibliographystyle{apalike}
\bibliography{library}


\clearpage
\newpage
\setcounter{page}{1}
\appendix

\setcounter{lemma}{0}
\setcounter{theorem}{0}
\setcounter{figure}{0}
\setcounter{table}{0}
\setcounter{assumption}{0}
\setcounter{equation}{0}
\renewcommand{\thefigure}{A\arabic{figure}}
\renewcommand{\thetable}{A\arabic{table}}
\renewcommand{\theremark}{A\arabic{remark}}
\renewcommand{\thelemma}{A\arabic{lemma}}
\renewcommand{\thetheorem}{A\arabic{theorem}}
\renewcommand{\theassumption}{A\arabic{assumption}}
\renewcommand{\theequation}{A\arabic{equation}}

\setcounter{definition}{0}
\renewcommand{\thedefinition}{A\arabic{definition}}


\section*{Appendix}

\section{Proof of Theorem \ref{thm:1}}\label{sec:proof}
\begin{proof}
Let $\|\cdot\|_{\infty}$ represent the maximum norm of a vector or matrix. Without loss of generality, assume $\|\c\|_2=1$. First, we derive the error rate for the whole $\widehat\bbeta\subatrel$ vector, which is above the parametric rate but useful in analyzing the second order error terms. Inspired by \cite{chen2016robust}, we expand the left side of (\ref{equ:dr:cross}) as 
\begin{equation}
\begin{split}
&\frac{1}{n}\sum_{k=1}^K\sum_{i\in\Isc_k}\omegahat\supfk(\X_i)\A_i\left\{Y_i-\mhat\supfk(\X_i)\right\}+\frac{1}{N}\sum_{i=n+1}^{N+n}\A_i\{\mhat(\X_i)-g(\A_i\trans\bbeta)\}\\
=&\frac{1}{n}\sum_{i=1}^n\bar\omega(\X_i)\A_i\left\{Y_i-\bar m(\X_i)\right\}+\frac{1}{N}\sum_{i=n+1}^{N+n}\A_i\{\bar m(\X_i)-g(\A_i\trans\bbeta)\}\\
&+\frac{1}{n}\sum_{k=1}^K\sum_{i\in\Isc_k}\{\omegahat\supfk(\X_i)-\bar\omega(\X_i)\}\A_i\{\mhat\supfk(\X_i)-\bar m(\X_i)\}\\
&+\frac{1}{n}\sum_{k=1}^K\sum_{i\in\Isc_k}\bar\omega(\X_i)\A_i\{\widehat m\supfk(\X_i)-\bar m(\X_i)\}-\frac{1}{N}\sum_{i=n+1}^{N+n}\A_i\{\widehat m(\X_i)-\bar m(\X_i)\}\\
&+\frac{1}{n}\sum_{k=1}^K\sum_{i\in\Isc_k}\{\omegahat\supfk(\X_i)-\bar\omega(\X_i)\}\A_i\left\{Y_i-\bar m(\X_i)\right\}\\
=:&\V(\bbeta)+\bDelta_a+\bDelta_b+\bDelta_c.
\end{split}
\label{equ:app:1}
\end{equation}
By Assumption \ref{asu:3}, independence between $\omegahat\supfk(\cdot)$ and data from $\Isc_k$ or data from the target population, and using the central limit theorem (CLT), we have that: for each $k$,
\begin{align*}
&\frac{K}{n}\sum_{i\in\Isc_k}\{\omegahat\supfk(\X_i)-\bar\omega(\X_i)\}^2-\rE_1\{\omegahat\supfk(\X)-\bar\omega(\X)\}^2=o_p(n^{-1/2});\\
&\frac{K}{n}\sum_{i\in\Isc_k}\{\widehat m\supfk(\X_i)-\bar m(\X_i)\}^2-\rE_1\{\widehat m\supfk(\X)-\bar m(\X)\}^2=o_p(n^{-1/2});\\
&\frac{1}{N}\sum_{i=n+1}^{N+n}\{\widehat m(\X_i)-\bar m(\X_i)\}^2-\rE_0\{\widehat m(\X)-\bar m(\X)\}^2=o_p(n^{-1/2})
\end{align*}
Also, by Assumption \ref{asu:3} and Assumption \ref{asu:1}, we have that: for each $k$,
\begin{align*}
&\rE_1\{\widehat \omega\supfk(\X)-\bar \omega(\X)\}^2=\rE_1\left[\bar \omega(\X)\left\{\frac{\widehat \omega\supfk(\X)}{\bar \omega(\X)}-1\right\}^2\right]\\
=&\rE_1\left[\bar\omega^2(\X)\left(\|\bPsi\|_2^2\|\widehat\balpha\supfk-\bar\balpha\|_2^2+\left\{\widehat h\supfk(\Z)-\bar h(\Z)\right\}^2+\|\bPsi\|_2^4\|\widehat\balpha\supfk-\bar\balpha\|_2^4+\left\{\widehat h\supfk(\Z)-\bar h(\Z)\right\}^4\right)\right]\\
\leq&\rE_1\left[\{\bar\omega^4(\X)+\|\bPsi\|_2^4+\|\bPsi\|_2^8+O_p(n^{-1})\}\right]\|\widehat\balpha\supfk-\bar\balpha\|_2^2+\{1+o_p(1)\}\rE_1\left[\bar\omega^2(\X)\{\widehat h\supfk(\Z)-\bar h(\Z)\}^2\right]\\
=&O_p\left(\rE_1\left[\bar\omega^2(\X)\{\widehat h\supfk(\Z)-\bar h(\Z)\}^2\right]+n^{-1}\right)=o_p(n^{-1/2}),
\end{align*}
and that each $j\in\{0,1\}$,
\begin{align*}
&\rE_j\{\widehat m\supfk(\X)-\bar m(\X)\}^2\\
=&\rE_1\Big[\breve g^2\{\bar m(\X)\}\left(\|\bPhi\|_2^2\|\widehat\bgamma\supfk-\bar\bgamma\|_2^2+\left\{\widehat r\supfk(\Z)-\bar r(\Z)\right\}^2\right)\\
&+C_L^2\left(\|\bPhi\|_2^4\|\widehat\bgamma\supfk-\bar\bgamma\|_2^4+\left\{\widehat r\supfk(\Z)-\bar r(\Z)\right\}^4\right)\Big]\\
=&O_p\left(\rE_1\left[\breve g^2\{\bar m(\X)\}\{\widehat r\supfk(\Z)-\bar r(\Z)\}^2\right]+n^{-1}\right)=o_p(n^{-1/2}).        
\end{align*}
Thus, we have $\frac{K}{n}\sum_{i\in\Isc_k}\{\omegahat\supfk(\X_i)-\bar\omega(\X_i)\}^2=o_p(n^{-1/2})$, $\frac{K}{n}\sum_{i\in\Isc_k}\{\widehat m\supfk(\X_i)-\bar m(\X_i)\}^2=o_p(n^{-1/2})$ and $\frac{1}{N}\sum_{i=n+1}^{N+n}\{\widehat m(\X_i)-\bar m(\X_i)\}^2=o_p(n^{-1/2})$. Combining these with Assumption \ref{asu:1}, we have that
\begin{align*}
\|\bDelta_a\|_{\infty}\leq&n^{-1}\max_{i}\|\A_i\|_{\infty}\sum_{k=1}^K\sum_{i\in\Isc_k}\{\omegahat\supfk(\X_i)-\bar\omega(\X_i)\}^2+\{\widehat m\supfk(\X_i)-\bar m(\X_i)\}^2=o_p(n^{-1/2});\\
\|\bDelta_b\|_{\infty}\leq&\max_{i}\|\A_i\|_{\infty}\left[n^{-1}\sum_{k=1}^K\sum_{i\in\Isc_k}\bar\omega^2(\X_i)\right]^{\frac{1}{2}}\left[n^{-1}\sum_{k=1}^K\sum_{i\in\Isc_k}\{\widehat m(\X_i)-\bar m(\X_i)\}^2\right]^{\frac{1}{2}}\\
&+\max_{i}\|\A_i\|_{\infty}\left[N^{-1}\sum_{i=n+1}^{N+n}\{\widehat m(\X_i)-\bar m(\X_i)\}^2\right]^{\frac{1}{2}}=o_p(n^{-1/4});\\
\|\bDelta_c\|_{\infty}\leq &\max_{i}\|\A_i\|_{\infty}\left[n^{-1}\sum_{k=1}^K\sum_{i\in\Isc_k}Y_i^2+\bar m^2(\X_i)\right]^{\frac{1}{2}}\left[n^{-1}\sum_{k=1}^K\sum_{i\in\Isc_k}\{\widehat\omega(\X_i)-\bar\omega(\X_i)\}^2\right]^{\frac{1}{2}}=o_p(n^{-1/4}).
\end{align*}
Thus, $\bbetahat\subatrel$ solves: $\V(\bbeta)+o_p(n^{-1/4})=\bzero$. Let the solution of $\rE\V(\bbeta)=\bzero$ be $\bar\bbeta$. When $\bar\omega(\cdot)=\bbw(\cdot)$, 
\begin{align*}
\rE\V(\bbeta)=&\rE_1\bbw(\X)\X\{Y-g(\A\trans\bbeta)\}+\left[\rE_1\bbw(\X)\{g(\A\trans\bbeta)-\bar m(\X)\}-\rE_{0}\{g(\A\trans\bbeta)-\bar m(\X)\}\right]\\
=&\rE_{0}\X\{Y-g(\A\trans\bbeta)\}+\bzero.
\end{align*}
As $\bar m(\cdot)=\mu(\cdot)$, $\rE\V(\bbeta)=\bzero+\rE_{0}\{\bar \mu(\X)-g(\A\trans\bbeta)\}$. Both cases lead to that $\bbeta_0$ solves $\rE\V(\bbeta)=\bzero$. So under Assumption \ref{asu:2}, we have $\bar\bbeta=\bbeta_0$. By Assumption \ref{asu:1}, $\V(\bbeta)$ is continuous differential on $\bbeta$. Then using Theorem 8.2 of \cite{pollard1990empirical}, we have $\|\widehat\bbeta\subatrel-\bbeta_0\|_2=o_p(n^{-1/4})=o_p(1)$.

Then we consider the asymptotic expansion of $\c\trans\bbetahat\subatrel$. Noting that $\bbetahat\subatrel$ is consistent for $\bbeta_0$, by Theorem 5.21 of \cite{van2000asymptotic}, we expand (\ref{equ:app:1}) with respect to $\c\trans\bbetahat\subatrel$ as:
\begin{equation}
\begin{split}
&\sqrt{n}(\c\trans\bbetahat\subatrel-\c\trans\bbeta_0)\\
=&\nnhalf\sum_{i=1}^n\bar\omega(\X_i)\c\trans\widehat\bJ_{\breve\bbeta}^{-1}\A_i\left\{Y_i-\bar m(\X_i)\right\}+\frac{\sqrt{\rho}}{\sqrt{N}}\sum_{i=n+1}^{N+n}\c\trans\widehat\bJ_{\breve\bbeta}^{-1}\A_i\{\bar m(\X_i)-g(\A_i\trans\bbeta_0)\}\\
&+\nnhalf\sum_{k=1}^K\sum_{i\in\Isc_k}\{\omegahat\supfk(\X_i)-\bar\omega(\X_i)\}\c\trans\widehat\bJ_{\breve\bbeta}^{-1}\A_i\left\{Y_i-\bar m(\X_i)\right\}\\
&+\nnhalf\sum_{k=1}^K\sum_{i\in\Isc_k}\bar\omega(\X_i)\c\trans\widehat\bJ_{\breve\bbeta}^{-1}\A_i\{\widehat m\supfk(\X_i)-\bar m(\X_i)\}-\frac{\nhalf}{N}\sum_{i=n+1}^{N+n}\c\trans\widehat\bJ_{\breve\bbeta}^{-1}\A_i\{\widehat m(\X_i)-\bar m(\X_i)\}\\
&+\nnhalf\sum_{k=1}^K\sum_{i\in\Isc_k}\c\trans\widehat\bJ_{\breve\bbeta}^{-1}\A_i\{\omegahat\supfk(\X_i)-\bar\omega(\X_i)\}\{\mhat\supfk(\X_i)-\bar m(\X_i)\}\\
=:&V+\Xi_1+\Xi_2+\Delta_3,
\end{split}   
\label{equ:app:2}
\end{equation}
where $\breve\bbeta$ is some vector lying between $\bbeta_0$ and $\bbetahat\subatrel$. First, we shall show that $\|\widehat\bJ_{\breve\bbeta}^{-1}-\bJ_{\bbeta_0}^{-1}\|_{\infty}=O_p(n^{-1/4})$. Since the dimensionality of $\A$, $d$ is fixed, we have
\[
\left\|\widehat\bJ_{\breve\bbeta}^{-1}-\bJ_{\bbeta_0}^{-1}\right\|_{\infty}=\left\|\widehat\bJ_{\breve\bbeta}^{-1}\bJ_{\bbeta_0}^{-1}(\widehat\bJ_{\breve\bbeta}-\bJ_{\bbeta_0})\right\|_{\infty}\leq d^3\left\|\widehat\bJ_{\breve\bbeta}^{-1}\right\|_{\infty}\left\|\bJ_{\bbeta_0}^{-1}\right\|_{\infty}\left\|\widehat\bJ_{\breve\bbeta}-\bJ_{\bbeta_0}\right\|_{\infty}.
\]
Denote by $\A_i=(A_{1i},\ldots,A_{di})\trans$. By Assumption \ref{asu:1} and CLT, there exists a constant $C>0$ such that for $j,\ell\in\{1,\ldots,d\}$,
\begin{align*}
&\left|N^{-1}\sum_{i=n+1}^{n+N}A_{ji}A_{\ell i}\dot g(\A_i\trans\breve\bbeta)-\rE_0A_{ji}A_{\ell i}\dot g(\A\trans_i\bbeta_0)\right|\\
\leq& \left|N^{-1}\sum_{i=n+1}^{n+N}A_{ji}A_{\ell i} \{\dot g(\A_i\trans\breve\bbeta)-\dot g(\A\trans_i\bbeta_0)\}\right|+\left|N^{-1}\sum_{i=n+1}^{n+N}A_{ji}A_{\ell i}\dot g(\A_i\trans\bbeta_0)-\rE_0A_{ji}A_{\ell i}\dot g(\A\trans_i\bbeta_0)\right|\\
\leq&\left|N^{-1}\sum_{i=n+1}^{n+N}|A_{ji}A_{\ell i}| C_L|\A_i\trans\breve\bbeta-\A\trans_i\bbeta_0|\right|+O_p(n^{-1/2})\leq C\|\widehat\bbeta\subatrel-\bbeta_0\|_2+O_p(n^{-1/2})=o_p(n^{-1/4}).
\end{align*}
Also noting that $\|\bJ_{\bbeta_0}^{-1}\|_{\infty}$ is bounded by Assumption \ref{asu:1}, we have 
\begin{equation}
\left\|\widehat\bJ_{\breve\bbeta}^{-1}-\bJ_{\bbeta_0}^{-1}\right\|_{\infty}=o_p(n^{-1/4}).
\label{equ:app:a1}
\end{equation}
Under Assumption \ref{asu:2}, and similar to the deduction above, the expectation of 
\[
\nnhalf\sum_{i=1}^n\bar\omega(\X_i)\A_i\left\{Y_i-\bar m(\X_i)\right\}+\frac{\sqrt{\rho}}{\sqrt{N}}\sum_{i=n+1}^{N+n}\A_i\{\bar m(\X_i)-g(\A_i\trans\bbeta_0)\}
\]
is $\bzero$. So by Assumption \ref{asu:1}, equation (\ref{equ:app:a1}), CLT and Slutsky's Theorem, we have that $V$ weakly converges to $N(0,\sigma^2)$ where $\sigma^2$ represents the asymptotic variance of $V$ and is order $1$. We then consider the remaining terms separately. First, we have 
\begin{equation}
\begin{split}
\Xi_1=&\nnhalf\sum_{k=1}^K\sum_{i\in\Isc_k}\bar\omega(\X_i)\c\trans\widehat\bJ_{\breve\bbeta}^{-1}\A_i\left[Y_i-g\{\bphi\trans\bar{\bgamma}+\bar r(\Z)\}\right]\left[\bpsi_i\trans(\widehat\balpha\supfk-\bar\balpha)+O_p(\{\bpsi_i\trans(\widehat\balpha\supfk-\bar\balpha)\}^2)\right]\\
&+\nnhalf\sum_{k=1}^K\sum_{i\in\Isc_k}\bar\omega(\X_i)\bkappa\subismbbetan\left[Y_i-g\{\bphi\trans\bar{\bgamma}+\bar r(\Z)\}\right]\Delta h\supfk(\z_j)\\
&+\nnhalf\sum_{k=1}^K\sum_{i\in\Isc_k}\bar\omega(\X_i)\c\trans(\widehat\bJ_{\breve\bbeta}^{-1}-\bJ_{\bbeta_0}^{-1})\A_i\left[Y_i-g\{\bphi\trans\bar{\bgamma}+\bar r(\Z)\}\right]\Delta h\supfk(\z_j)\\
=:&U_1+\Delta_{11}+\Delta_{12},
\end{split}    
\label{equ:app:3}
\end{equation}
where $\Delta h\supfk(\z_j)=\widehat{h}\supfk(\Z_i)-\bar{h}(\Z_i)+O_p(\{\widehat{h}\supfk(\Z_i)-\bar{h}(\Z_i)\}^2)$. Recall that 
\[
\bzeta_{\alpha}=\rE_1\bar\omega(\X)\bkappa\subsmbbetan\left[Y-g\{\bphi\trans\bar{\bgamma}+\bar r(\Z)\}\right]\bpsi.
\]
Again using (\ref{equ:app:a1}) and Assumption \ref{asu:1}, we have that
\[
n^{-1}\sum_{k=1}^K\sum_{i\in\Isc_k}\bar\omega(\X_i)\c\trans\widehat\bJ_{\breve\bbeta}^{-1}\A_i\left[Y_i-g\{\bphi\trans\bar{\bgamma}+\bar r(\Z)\}\right]\xrightarrow{p}\bzeta_{\alpha}.
\]
Combining this with Assumption \ref{asu:1}, Assumption \ref{asu:3} that $\sqrt{n}(\widehat\balpha\supfk-\bar\balpha)$ is asymptotic normal with mean $0$ and covariance of order $1$, and using Slutsky's Theorem, we have that $U_1$ is asymptotically equivalent with $\sqrt{n}\bzeta_{\alpha}\trans(\widehat\balpha-\bar\balpha)$, which weakly converges to normal distribution with mean $0$ and variance of order $1$.

For $\Delta_{11}$, by Assumption \ref{asu:2}, the moment condition:
\[
\rE_1\left[\bar\omega(\X)\bkappa\subsmbbetan\left(Y-g\{\bPhi\trans\bar{\bgamma}+\bar r(\Z)\}\right)\Big|\Z\right]=0
\]
holds because under Assumption \ref{asu:2}(i), both limiting parameters $\omega^*(\cdot)=\bar\omega(\cdot)=\omega(\cdot)$ and $\bar r(\cdot)$ solves (\ref{equ:mom:1}) while under \ref{asu:2}(ii), $\rE_1[Y|\X]=g\{\bPhi\trans\bar{\bgamma}+\bar r(\Z)\}$, leading to
\[
\rE_1\left[\bar\omega(\X)\bkappa\subsmbbetan\left(Y-g\{\bPhi\trans\bar{\bgamma}+\bar r(\Z)\}\right)\Big|\X\right]=0.
\]
Combining this with the fact that $\widehat{h}\supfk(\cdot)$ is independent of the data in $\Isc_k$ due to the use of cross-fitting, we have $\rE_1\Delta_{11}=\rE_1[\Delta_{11}\mid\widehat{h}\supfk(\cdot)]=0+n^{1/2}O_p(\{\widehat{h}\supfk(\Z_i)-\bar{h}(\Z_i)\}^2)$. By Assumptions \ref{asu:1} and \ref{asu:3}(ii), we have that
\begin{align*}
&{\rm Var}_1\left(\bar\omega(\X_i)\bkappa\subismbbetan\left[Y_i-g\{\bphi\trans\bar{\bgamma}+\bar r(\Z)\}\right]\{\widehat{h}\supfk(\Z_i)-\bar{h}(\Z_i)\}\Big|\widehat{h}\supfk(\cdot)\right)\\
=&O(\rE_1[\bar\omega^2(\X_i)+Y_i^2+\bar m^2(\X_i)])\cdot o_p(1)=o_p(1),
\end{align*}
where ${\rm Var}_1$ and ${\rm Var}_0$ represent the variance operator of the source and target population respectively. Then by CLT and Assumption \ref{asu:3}(ii), we have that
\[
\Delta_{11}=\left(\Delta_{11}-\rE_1[\Delta_{11}|\widehat{h}\supfk(\cdot)]\right)+\rE_1[\Delta_{11}|\widehat{h}\supfk(\cdot)]=o_p(1)+n^{1/2}O_p(\{\widehat{h}\supfk(\Z_i)-\bar{h}(\Z_i)\}^2)=o_p(1).
\]
For term $\Delta_{12}$, by (\ref{equ:app:a1}) and Assumptions \ref{asu:1} and \ref{asu:3}, there exists constant $C_{12}>0$ such that 
\begin{align*}
|\Delta_{12}|\leq  C_{12}\max_i\|\A_i\|_{\infty}\left\|\widehat\bJ_{\breve\bbeta}^{-1}-\bJ_{\bbeta_0}^{-1}\right\|_{\infty}\left[n^{-1}\sum_{k=1}^K\sum_{i\in\Isc_k}\bar\omega^2(\X_i)\{\widehat{h}\supfk(\Z_i)-\bar{h}(\Z_i)\}^2\right]^{\frac{1}{2}}+o_p(1)=o_p(1).
\end{align*}
Therefore, we come to that $\Xi_1$ is asymptotically equivalent with $\sqrt{n}\bzeta_{\alpha}\trans(\widehat\balpha-\bar\balpha)$. Similarly, we write the term $\Xi_2$ as 
\begin{equation}
\begin{split}
\Xi_2=&\nnhalf\sum_{k=1}^K\sum_{i\in\Isc_k}\bar{\omega}(\X_i)\c\trans\widehat\bJ_{\breve\bbeta}^{-1}\A_i\breve g\{\bar m(\X_i)\}\left[\bphi_i\trans(\bgammahat\supfk-\bar\bgamma)+O_p(\{\bphi_i\trans(\bgammahat\supfk-\bar\bgamma)\}^2)\right]\\
&-\frac{\nhalf}{N}\sum_{i=n+1}^{N+n}\c\trans\widehat\bJ_{\breve\bbeta}^{-1}\A_i\breve g\{\bar m(\X_i)\}\left[K^{-1}\sum_{k=1}^K\bphi_i\trans(\bgammahat\supfk-\bar\bgamma)+O_p(\{\bphi_i\trans(\bgammahat\supfk-\bar\bgamma)\}^2)\right]\\
&+\nnhalf\sum_{k=1}^K\sum_{i\in\Isc_k}\bar{\omega}(\X_i)\bkappa\subismbbetan\breve g\{\bar m(\X_i)\}\Delta r\supfk(\Z_i)-\frac{\nhalf}{N}\sum_{i=n+1}^{N+n}\bkappa\subismbbetan\breve g\{\bar m(\X_i)\}\Delta r(\Z_i)\\
&+\nnhalf\sum_{k=1}^K\sum_{i\in\Isc_k}\bar{\omega}(\X_i)\c\trans\left[\widehat\bJ_{\breve\bbeta}^{-1}-\bJ_{\bbeta_0}^{-1}\right]\A_i\breve g\{\bar m(\X_i)\}\Delta r\supfk(\Z_i)\\
&-\frac{\nhalf}{N}\sum_{i=n+1}^{N+n}\c\trans\left[\widehat\bJ_{\breve\bbeta}^{-1}-\bJ_{\bbeta_0}^{-1}\right]\A_i\breve g\{\bar m(\X_i)\}\Delta r(\Z_i)\\
=:&U_2+\Delta_{21}+\Delta_{22},
\end{split}    
\label{equ:app:4}
\end{equation}
where $\Delta r\supfk(\Z_i)=\rhat\supfk(\Z_i)-\bar r(\Z_i)+O_p(\{\rhat\supfk(\Z_i)-\bar r(\Z_i)\}^2)$, $\Delta r(\Z_i)=K^{-1}\sum_{k=1}^K\Delta r\supfk(\Z_i)$, $U_2$ represents the difference of the first two terms, and $\Delta_{22}$ represents the difference of the last two terms. Similar to $U_1$, by (\ref{equ:app:a1}) and Assumption \ref{asu:1},
\[
\frac{1}{n}\sum_{k=1}^K\sum_{i\in\Isc_k}\bar{\omega}(\X_i)\c\trans\widehat\bJ_{\breve\bbeta}^{-1}\A_i\breve g\{\bar m(\X_i)\}\bphi_i-\frac{1}{N}\sum_{i=n+1}^{N+n}\c\trans\widehat\bJ_{\breve\bbeta}^{-1}\A_i\breve g\{\bar m(\X_i)\}\bphi_i\xrightarrow{p}\bzeta_{\gamma}.
\]
Again, combining this with Assumptions \ref{asu:1} and Assumption \ref{asu:3}, and using Slutsky's Theorem, we have that $U_2$ is asymptotically equivalent with $\sqrt{n}\bzeta_{\gamma}\trans(\widehat\bgamma-\bar\bgamma)$, which weakly converges to normal distribution with mean $0$ and variance of order $1$.

For $\Delta_{21}$, by Assumptions \ref{asu:2} and \ref{asu:3}, as well as the use of cross-fitting, we have that
\begin{align*}
&\rE_1\left(\frac{1}{n}\sum_{k=1}^K\sum_{i\in\Isc_k}\bar{\omega}(\X_i)\bkappa\subismbbetan\breve g\{\bar m(\X_i)\}\Delta r\supfk(\Z_i)\right)-\rE_{0}\left(\frac{1}{N}\sum_{i=n+1}^{N+n}\bkappa\subismbbetan\breve g\{\bar m(\X_i)\}\Delta r\supfk(\Z_i)\right)=o_p(n^{-1/2}).
\end{align*}
Here, we follow the same idea as that for $\Delta_{11}$: if Assumption \ref{asu:2}(i) holds, we have $\bar\omega(\cdot)=\bbw(\cdot)$ and
\[
\rE_1\left[\exp\{\bPsi\trans\bar{\balpha}+\bar{h}(\Z)\}\bkappa\subsmbbetan\breve{g}\{\bar m(\X)\}f(\X)\right]=\rE_{0}\left[\bkappa\subsmbbetan\breve{g}\{\bar m(\X)\}f(\X)\right]
\]
holds for all measurable function of $\X$, $f(\cdot)$; when Assumption \ref{asu:2}(ii) holds, we have that $m^*(\cdot)=\bar m(\cdot)=\mu(\cdot)$ and thus $\bar h(\cdot)$ solves (\ref{equ:mom:2}). Also note that
\begin{align*}
&{\rm Var}_1\left(\bar{\omega}(\X_i)\bkappa\subismbbetan\breve g\{\bar m(\X_i)\}\{\rhat\supfk(\Z_i)-\bar r(\Z_i)\}\Big|\rhat\supfk(\cdot)\right)\\
=&O(\rE_1[\bar\omega^2(\X_i)+\breve g^2\{\bar m(\X_i)\}])\cdot o_p(1)=o_p(1);\\
&{\rm Var}_0\left(\bkappa\subismbbetan\breve{g}\{\bar m(\X_i)\}\{\rhat\supfk(\Z_i)-\bar r(\Z_i)\}\Big|\rhat\supfk(\cdot)\right)=O(\rE_1\breve g^2\{\bar m(\X_i)\})\cdot o_p(1)=o_p(1);
\end{align*}
Then similar to $\Delta_{12}$, we come to $\Delta_{22}=o_p(1)$. Thus, the term $\Xi_2$ is asymptotically equivalent with $\sqrt{n}\bzeta_{\gamma}\trans(\widehat\bgamma-\bar\bgamma)$, which weakly converges to normal distribution with mean $0$ and variance of order $1$. 

Finally, we consider $\Delta_3$ in (\ref{equ:app:2}). By Assumption \ref{asu:1}, the boundness of $|\c\trans\widehat\bJ_{\breve\bbeta}^{-1}\A_i|$ and our derived bounds for $n^{-1}\sum_{k=1}^K\sum_{i\in\Isc_k}\{\omegahat\supfk(\X_i)-\bar\omega(\X_i)\}^2$ and $n^{-1}\sum_{k=1}^K\sum_{i\in\Isc_k}\{\widehat m\supfk(\X_i)-\bar m(\X_i)\}^2$,
\begin{align*}
|\Delta_3|=&O\left(\nnhalf\sum_{k=1}^K\sum_{i\in\Isc_k}|\omegahat\supfk(\X_i)-\bar\omega(\X_i)||\mhat\supfk(\X_i)-\bar m(\X_i)|\right)\\
\leq&\sqrt{n}O\left(\left[n^{-1}\sum_{k=1}^K\sum_{i\in\Isc_k}\{\omegahat\supfk(\X_i)-\bar\omega(\X_i)\}^2\right]^{\frac{1}{2}}\left[n^{-1}\sum_{k=1}^K\sum_{i\in\Isc_k}\{\widehat m\supfk(\X_i)-\bar m(\X_i)\}^2\right]^{\frac{1}{2}}\right)=o_p(1).
\end{align*}
Combining this with the asymptotic properties derived for $V$, $\Xi_1$ and $\Xi_2$ and the expansion (\ref{equ:app:2}), we finish the proof for the asymptotic expansion and distribution of $\sqrt{n}(\c\trans\bbetahat\subatrel-\c\trans\bbeta_0)$.
\end{proof}

\newpage

\section{Additional assumptions and justification of Proposition \ref{prop:1}}\label{sec:app:kern}
In this section, we present the additional assumptions and justification for Proposition \ref{prop:1} that establishes the convergence rates and asymptotic behaviour of our mainly studied nuisance estimators defined in Section \ref{sec:method:spec}. Our results are largely based on existing literature of local regression and sieve like \cite{fan1995local}, \cite{shen1997methods}, \cite{carroll1998local} and \cite{chen2007large}.

Denote by $G(x)=\int^x_{-\infty} g(t)dt$. Let $\Lambda_{\alpha^*}$, $\Lambda_{\gamma^*}$, $\Lambda_{h^*}$, $\Lambda_{r^*}$, $\Lambda_{\bar h}$ and $\Lambda_{\bar r}$ represent the parameter space of $\alpha^*$, $\gamma^*$, $h^*$, $r^*$, $\bar h$ and $\bar r$ respectively. 
Let $\Zsc$ be the domain of $\Z\in\mathbb{R}^{p_{\z}}$ and $\Csc^k(\Zsc)$ represent all the $k$-times differentiable continuous functions on $\Zsc$. The H\"{o}lder (or $\nu$-smooth) class $\Sigma(\nu,L)$ is defined as the set of functions $f\in\Csc^{[\nu]}(\Zsc)$ with its $[\nu]$-times derivative satisfying
\[
\sup_{\z_1,\z_2\in\Zsc}\frac{\|f^{([\nu])}(\z_1)-f^{([\nu])}(\z_2)\|_2}{\|\z_1-\z_2\|_2}\leq L.
\]


\begin{assumption}
\label{asu:a1}
(i) $\bphi$, $\bpsi$ and $\Z$ have compact domain and continuous differentiable probability density functions (as given for discrete variables).

~\\
(ii) There exists $C_1>0$ that for all $\z\in\Zsc$,
\[
\|\balpha^*\|_{\infty},\|\bgamma^*\|_{\infty},|h^*(\z)|,|r^*(\z)|,|\bar h(\z)|,|\bar r(\z)|\leq C_1.
\] 
(iii) There exists $C_2>0$ such that
\begin{align*}
C_2^{-1}\leq&\frac{\frac{\partial}{\partial\tau} \rE_1 \exp\{\bpsi\trans[\balpha_1+\tau(\balpha_2-\balpha_1)]+h_1(\Z)+\tau[h_2(\Z)-h_1(\Z)]\}}{\|\balpha_1-\balpha_2\|_2^2+\rE_1[h_1(\Z)-h_2(\Z)]^2}\leq C_2;\\
C_2^{-1}\leq&\frac{\frac{\partial}{\partial\tau} \rE_1 G\{\bphi\trans[\bgamma_1+\tau(\bgamma_2-\bgamma_1)]+r_1(\Z)+\tau[r_2(\Z)-r_1(\Z)]\}}{\|\bgamma_1-\bgamma_2\|_2^2+\rE_1[r_1(\Z)-r_2(\Z)]^2}\leq C_2,
\end{align*}
for any $\tau\in[0,1]$, $\balpha_1,\balpha_2\in\Lambda_{\alpha^*}$, $h_1,h_2\in\Lambda_{h^*}$, $\bgamma_1,\bgamma_2\in\Lambda_{\gamma^*}$, and $r_1,r_2\in\Lambda_{r^*}$.

~\\
(iv) It holds that $\bkappa_{\bbeta_0}\geq 0$ with probability $1$. There exists $C_3>0$ that for all $\z\in\Zsc$,
\begin{align*}
&C_3^{-1}\leq\left|h^{-p_{\z}}\rE_1K_h(\Z-\z){\omega}^*(\X)\bkappa_{\bbeta_0}\dot g\left\{\bphi\trans\bar\bgamma+\bar r(\z)\right\}\right|\leq C_3;\\ 
&C_3^{-1}\leq\left|h^{-p_{\z}}\rE_1K_h(\Z-\z)\exp(\bpsi\trans\bar{\balpha})\bkappa_{\bbeta_0}\breve{g}\{m^*(\X)\}\exp\{\bar h(\z)\}\right|\leq C_3.
\end{align*}
\end{assumption}

\begin{assumption}
There exists $\nu,L>0$ such that all population-level nonparametric components $h^*(\z),r^*(\z),\bar h(\z)$ and $\bar r(\z)$ belong to the H\"{o}lder class $\Sigma(\nu,L)$ with the degree of smoothness $\nu$ satisfying $\nu>p_{\z}$. 
\label{asu:a2}
\end{assumption}

\begin{assumption}[Specification of the sieve and kernel functions]
(i) The basis function $\b(\Z)$ is taken as the tensor product of $\b_{j}(Z_j)$ for $j=1,2,\ldots,p_{\z}$, where each $\b_{j}(Z_j)$ is the Hermite polynomial basis of the univariate $Z_j$ with its order $s\asymp n^{1/(p_z+\nu)}$. (ii) The kernel function $K$ is symmetric, bounded, and of order $[\nu]$ and the bandwidth $h\asymp n^{-1/(p_z+2\nu)}$. The tuning parameters $\lambda_1,\lambda_2=o(n^{-1/2})$.

\label{asu:a3}
\end{assumption}

\begin{remark}
Similar to Assumption \ref{asu:1} in the main paper, Assumptions \ref{asu:a1}(i) and \ref{asu:a1}(ii) are used to regular the distribution of $\X$ and the parameter spaces. Assumption \ref{asu:a1}(iii) is in a similar spirit of Condition 4.5 in \cite{chen2007large}, used to control the asymptotic variance of $\sqrt{n}(\widetilde\balpha\supfk-\balpha^*)$ and $\sqrt{n}(\widetilde\balpha\supfk-\balpha^*)$. Assumption \ref{asu:a1}(iv) requires the weighting term $\bkappa_{\bbeta_0}$ to be positive-definite to ensure the regularity of the calibration equations. As we remark in Remark \ref{rem:sign:split}, this assumption can be granted by splitting the samples by the sign of $\bkappa_{\widetilde\bbeta}$ when it is not always positive or always negative. Assumption \ref{asu:a2} imposes the common smoothness conditions on the nuisance nonparametric components that are also used in semiparametric inference existing literature like \cite{rothe2015semiparametric} and \cite{chakrabortty2018efficient}. In Assumption \ref{asu:a3}, we choose the order of sieve of the preliminary nuisance estimators to be under-smoothed optimal since $\sqrt{n}$-consistency of the parametric part in these models are required. While the bandwidth $h$ used in the calibrated estimating equation (\ref{equ:kern:mom}) can be rate-optimal since we do not need to estimate the parametric components in this step.

\end{remark}

\begin{proof}[Proof of Proposition \ref{prop:1}]
Since we simply pick $\widehat\balpha\supfk=\widetilde\balpha\supfk$ and $\widehat\bgamma\supfk=\widetilde\bgamma\supfk$ in Section \ref{sec:method:spec}, Assumptions \ref{asu:1} and \ref{asu:a1}--\ref{asu:a3} are sufficient for Assumption \ref{asu:3}(i) by Lemma \ref{lemma:A3}(b) presented and justified in this section. And Assumption \ref{asu:3}(ii) is directly given by Lemma \ref{lemma:A4} that is proved based on Lemmas \ref{lemma:A1}--\ref{lemma:A3}.

\end{proof}

Lemma \ref{lemma:A1} establishes the desirable convergence properties of the preliminary nuisance estimators based on the existing analysis of sieve M-estimation \citep{shen1997methods,chen2007large}.

\begin{lemma}[\citep{shen1997methods,chen2007large}]
Under Assumptions \ref{asu:1} and \ref{asu:a1}--\ref{asu:a3}, the preliminary nuisance estimators solved from equations (\ref{equ:pre:nui:1}) and (\ref{equ:pre:nui:2}) satisfy that:
\\
(a) For $j\in\{0,1\}$,
\begin{align*}
&\rE_1\{\widetilde r\supfk(\Z)-r^*(\Z)\}^2+\rE_j\{\widetilde h\supfk(\Z)-h^*(\Z)\}^2=o_p(n^{-1/2});\\
&\sup_{\z\in\mathcal{Z}}|\widetilde r\supfk(\z)-r^*(\z)|+|\widetilde  h\supfk(\z)-h^*(\z)|=o_p(1);
\end{align*}
(b) $\sqrt{n}(\widetilde\balpha\supfk-\balpha^*)$ and $\sqrt{n}(\widetilde\balpha\supfk-\balpha^*)$ weakly converge to gaussian distributuon with mean zero and finite variance.
\label{lemma:A1}
\end{lemma}

\begin{proof}
We based on Theorem 3.5 of \cite{chen2007large} to show (a) of Lemma \ref{lemma:A1}. First, note that for both preliminary nuisance models, Conditions 3.9, 3.10, 3.11 and 3.13 of \cite{chen2007large} are implied by Assumptions \ref{asu:1}, \ref{asu:a1}(i) and \ref{asu:a1}(ii). Their Condition 3.12 is implied by Assumption \ref{asu:a1}(iii). Then by their Theorem 3.5, it holds that
\begin{align*}
&\|\widetilde\bgamma\supfk-\bgamma^*\|_2^2+\rE_1\{\widetilde r\supfk(\Z)-r^*(\Z)\}^2=O_p\left(\frac{k_n}{n}+\rho_{2n}^2\right);\\
&\|\widetilde\balpha\supfk-\balpha^*\|_2^2+\rE_1\{\widetilde h\supfk(\Z)-h^*(\Z)\}^2=O_p\left(\frac{k_n}{n}+\rho_{2n}^2\right),
\end{align*}
where $k_n$ and $\rho_{2n}^2$ respectively characterize the variance and approximation bias of sieve to be specified as follows. Inspired by Proposition 3.6 of \cite{chen2007large}, under our Assumptions \ref{asu:a2} and \ref{asu:a3}(i), the specific rate of $k_n$ and $\rho_{2n}^2$ is given by
\[
k_n\asymp s^{p_z},\quad \rho_{2n}\asymp s^{-\nu},\mbox{ where $s$ is the order of each $\b_j(Z_j)$}.
\]
Then by Assumption \ref{asu:a2} that $\nu>p_{\z}$ and Assumption \ref{asu:a3}(i) that $s\asymp n^{1/(p_z+\nu)}$, we have 
\begin{align*}
&\|\widetilde\bgamma\supfk-\bgamma^*\|_2^2+\rE_1\{\widetilde r\supfk(\Z)-r^*(\Z)\}^2=o_p(n^{-1/2});\\
&\|\widetilde\balpha\supfk-\balpha^*\|_2^2+\rE_1\{\widetilde h\supfk(\Z)-h^*(\Z)\}^2=o_p(n^{-1/2}).
\end{align*}
Similarly, it is not hard to justify that our Assumptions \ref{asu:1} and \ref{asu:a1}--\ref{asu:a3} imply Conditions 3.1, 3.2, 3.4 and 3.5M of \cite{chen2007large}, which are sufficient for the consistency of sieve M-estimation according to their Remark 3.3, i.e.,
\[
\sup_{\z\in\mathcal{Z}}|\widetilde r\supfk(\z)-r^*(\z)|+|\widetilde  h\supfk(\z)-h^*(\z)|=o_p(1).
\]
So we finish proving (a) of Lemma \ref{lemma:A1}. 

Next, we prove (b) based on (a) and using Theorem 4.3 of \cite{chen2007large} (or early works like \cite{shen1997methods}). Their Conditions 4.1(iii) and 4.4 are as given in our standard non-linear M-estimation case. Since ``$f(\theta)$" in \cite{chen2007large} are simply the parametric parts $\bgamma$ or $\balpha$ in our case, their Conditions 4.1(i) and 4.2(ii) are trivially satisfied. Their Condition 4.5 is implied by our Assumption \ref{asu:a1}(iii) that actually indicates $\sqrt{n}(\widetilde\balpha\supfk-\balpha^*)$ and $\sqrt{n}(\widetilde\balpha\supfk-\balpha^*)$ will have bounded asymptotic variance. And their Conditions 4.2' and 4.3' are implied by Assumption \ref{asu:a1}(i) and the continuity of the link function $g$. Therefore, we can combine our Lemma \ref{lemma:A1}(a) and Theorem 4.3 of \cite{chen2007large} to finishe the proof of Lemma \ref{lemma:A1}(b).

\end{proof}

Using Lemma \ref{lemma:A1} and that at least one nuisance model is correctly specified (i.e., Assumption \ref{asu:2}), Lemma \ref{lemma:A2} establishes the $o_p(n^{-1/4})$ convergence of the preliminary estimator $\widetilde\bbeta\supfk$ to the true $\bbeta_0$.
\begin{lemma}
\label{lemma:A2}
Under Assumptions \ref{asu:1}, \ref{asu:2} and \ref{asu:a1}--\ref{asu:a3},
\[
\rE_j\{\widetilde m\supfk(\X)-m^*(\X)\}^2+\rE_1\{\widetilde \omega\supfk(\X)-\omega^*(\X)\}^2+\|\widetilde\bbeta\supfk-\bbeta_0\|_2^2=o_p(n^{-1/2}).
\]
\end{lemma}
\begin{proof}
It immediately follows from Lemma \ref{lemma:A1} that
\[
\rE_j\{\widetilde m\supfk(\X)-m^*(\X)\}^2+\rE_1\{\widetilde \omega\supfk(\X)-\omega^*(\X)\}^2=o_p(n^{-1/2}).
\]
Then $\|\widetilde\bbeta\supfk-\bbeta_0\|_2^2=o_p(n^{-1/2})$ can be proved by following the same proof procedures in Theorem \ref{thm:1} for analyzing the terms defined in (\ref{equ:app:1}).
\end{proof}

For each $\z\in\mathcal{Z}$, let the estimators $\breve r\supfk(\z)$ and $\breve h\supfk(\z)$ respectively solve:
\begin{equation}
\begin{split}
&\frac{K}{n(K-1)h^{p_{\z}}}\sum_{i\in \Isc\subminusk} K_h(\Z_i-\z)\omega^*(\X_i)\bkappa\subismbbetan\left[Y_i-g\left\{\bphi_i\trans\bar{\bgamma}+r(\z)\right\}\right]=0;\\
&\frac{K}{n(K-1)h^{p_{\z}}}\sum_{i\in\Isc\subminusk}K_h(\Z_i-\z)\exp(\bpsi_i\trans\bar{\balpha})\bkappa\subismbbetan\breve{g}\{m^*(\X_i)\}\exp\{h(\z)\}\\
=&\frac{1}{Nh^{p_{\z}}}\sum_{i=n+1}^{n+N}K_h(\Z_i-\z)\bkappa\subismbbetan\breve{g}\{ m^*(\X_i)\},
\end{split}
\label{equ:kern:mom:ideal}    
\end{equation}
i.e. the ``oracle" version of the estimating equations in (\ref{equ:kern:mom}), obtained by replacing all the preliminary estimators plugged in (\ref{equ:kern:mom}) with their limits (true values). Also recall that $\bar h(\z)$ and $\bar r(\z)$ are defined as the solutions to equations (\ref{equ:mom:1}) and (\ref{equ:mom:2}).

We introduce Lemma \ref{lemma:A3} to give the consistency $o_p(n^{-1/4})$ convergence of $\breve h\supfk(\z)$ and $\breve r\supfk(\z)$ to $\bar h(\z)$ and $\bar r(\z)$, as a standard result of the higher--order kernel (or local polynomial) estimating equation \citep{fan1995local}.

\begin{lemma}
\label{lemma:A3}
Under Assumptions \ref{asu:1}, \ref{asu:2} and \ref{asu:a1}--\ref{asu:a3}, 
\begin{align*}
&\rE_1\{\breve r\supfk(\Z)-\bar r(\Z)\}^2+\rE_1\{\breve h\supfk(\Z)-\bar h(\Z)\}^2=o_p(n^{-1/2});\\
&\sup_{\z\in\mathcal{Z}}|\breve r\supfk(\z)-\bar r(\z)|+|\breve h\supfk(\z)-\bar h(\z)|=o_p(1).
\end{align*}
\end{lemma}

\begin{proof}
By Assumption \ref{asu:2}, at least one nuisance model is correctly specified. When the importance weighting model is correct, $w^*(\x)=\bar w(\x)=\bbw(\x)$. So the first equation of (\ref{equ:kern:mom:ideal}) is (asymptotically) valid for $\bar r(\Z)$ that solves (\ref{equ:mom:1}). Also, since $\bbw(\x)=\exp(\psi\trans\balpha_0+h_0(\z))$ and $\bar\balpha=\balpha_0$ when the importance weighting model is correct, the second equation of (\ref{equ:kern:mom:ideal}) is valid for $\bar h(\z)=h_0(\z)$ that solves (\ref{equ:mom:2}). So both equations in (\ref{equ:kern:mom:ideal}) are valid. Similarly, this also holds when the imputation model is correct. Then by Assumptions \ref{asu:1}, and \ref{asu:a1}--\ref{asu:a3} and following Appendix A of \cite{fan1995local}, we can derive that $\sup_{\z\in\mathcal{Z}}|\breve r\supfk(\z)-\bar r(\z)|+|\breve h\supfk(\z)-\bar h(\z)|=o_p(1)$ and 
\[
\rE_1\{\breve r\supfk(\Z)-\bar r(\Z)\}^2+\rE_1\{\breve h\supfk(\Z)-\bar h(\Z)\}^2=O_p\left(\frac{1}{nh^{p_{\z}}}+h^{2\nu}\right)=o_p(n^{-1/2}),
\]
as the standard consistency and convergence results of kernel smoothing. 

Note that \citep{fan1995local} studied the local polynomial regression approach that is not exactly the same as our used $[\nu]$-th order kernel; see Assumption \ref{asu:a3}(ii). While the derivation of these two approaches are basically the same due to the orthogonality between a $[\nu]$-th order kernel function and the polynomial functions of the order up to $[\nu]$.

\end{proof}

Finally, we come to Lemma \ref{lemma:A4} for the asymptotic properties of $\widehat r\supfk(\Z)$ and $\widehat h\supfk(\Z)$.
\begin{lemma}
\label{lemma:A4}
Under Assumptions \ref{asu:1}, \ref{asu:2} and \ref{asu:a1}--\ref{asu:a3}, the calibrated nuisance estimators satisfy:
\begin{align*}
&\rE_1\{\widehat r\supfk(\Z)-\bar r(\Z)\}^2+\rE_1\{\widehat h\supfk(\Z)-\bar h(\Z)\}^2=o_p(n^{-1/2});\\
&\sup_{\z\in\mathcal{Z}}|\widehat r\supfk(\z)-\bar r(\z)|+|\widehat h\supfk(\z)-\bar h(\z)|=o_p(1).
\end{align*}
\end{lemma}

\begin{proof}
We compare the estimating equations in (\ref{equ:kern:mom}) with those in (\ref{equ:kern:mom:ideal}) to analyze the additional errors incurred by the preliminary estimators in (\ref{equ:kern:mom}). 
By Assumption \ref{asu:1} and equation (\ref{equ:app:a1}) derived in the proof of Theorem \ref{thm:1}, we have that for each $\z$,
\begin{align*}
0=&\frac{K}{n(K-1)h^{p_{\z}}}\sum_{i\in \Isc\subminusk} K_h(\Z_i-\z)\widetilde{\omega}\supfk(\X_i)\c\trans\bJhat_{\widetilde{\bbeta}\supfk}^{-1}\A_i\left[Y_i-g\left\{\bphi_i\trans\widehat{\bgamma}\supfk+\widehat r\supfk(\z)\right\}\right]\\
=&\frac{K}{n(K-1)h^{p_{\z}}}\sum_{i\in \Isc\subminusk} K_h(\Z_i-\z){\omega}^*(\X_i)\bkappa\subismbbetan\left[Y_i-g\left\{\bphi_i\bar\bgamma+\widehat r\supfk(\z)\right\}\right]\\
&+\frac{K}{n(K-1)h^{p_{\z}}}\sum_{i\in \Isc\subminusk} K_h(\Z_i-\z){\omega}^*(\X_i)\bkappa\subismbbetan\left[g\left\{\bphi_i\trans\bar\bgamma+\widehat r\supfk(\z)\right\}-g\left\{\bphi_i\trans\widehat{\bgamma}\supfk+\widehat r\supfk(\z)\right\}\right]\\
&+\frac{K}{n(K-1)h^{p_{\z}}}\sum_{i\in \Isc\subminusk} K_h(\Z_i-\z){\omega}^*(\X_i)\c\trans\left[\bJhat_{\widetilde{\bbeta}\supfk}^{-1}-\bJ_{\bbeta_0}^{-1}\right]\A_i\left[Y_i-g\left\{\bphi_i\trans\widehat{\bgamma}\supfk+\widehat r\supfk(\z)\right\}\right]\\
&+\frac{K}{n(K-1)h^{p_{\z}}}\sum_{i\in \Isc\subminusk} K_h(\Z_i-\z)\{\widetilde{\omega}\supfk(\X_i)-{\omega}^*(\X_i)\}\c\trans\bJhat_{\widetilde{\bbeta}\supfk}^{-1}\A_i\left[Y_i-g\left\{\bphi_i\trans\widehat{\bgamma}\supfk+\widehat r\supfk(\z)\right\}\right]\\
=&\frac{K}{n(K-1)h^{p_{\z}}}\sum_{i\in \Isc\subminusk} K_h(\Z_i-\z){\omega}^*(\X_i)\bkappa\subismbbetan\left[Y_i- g\left\{\bphi_i\trans\bar\bgamma+\widehat r\supfk(\z)\right\}\right]\\
&+O_p\left(\left[\rE_1\{\widetilde \omega\supfk(\X)-\omega^*(\X)\}^2\right]^{\frac{1}{2}}+\|\widetilde\bbeta\supfk-\bbeta_0\|_2+\|\widehat\bgamma\supfk-\bar \bgamma\|_2+n^{-1/2}\right)\\
=&\frac{K}{n(K-1)h^{p_{\z}}}\sum_{i\in \Isc\subminusk} K_h(\Z_i-\z){\omega}^*(\X_i)\bkappa\subismbbetan\left[Y_i- g\left\{\bphi_i\trans\bar\bgamma+\widehat r\supfk(\z)\right\}\right]+o_p(n^{-1/4}),
\end{align*}
Comparing this with the estimating equation (\ref{equ:kern:mom:ideal}) for $\breve r\supfk(\cdot)$, we have:
\begin{align*}
\frac{K}{n(K-1)h^{p_{\z}}}\sum_{i\in \Isc\subminusk} K_h(\Z_i-\z){\omega}^*(\X_i)\bkappa\subismbbetan\left[g\left\{\bphi_i\trans\bar\bgamma+\breve r\supfk(\z)\right\}- g\left\{\bphi_i\trans\bar\bgamma+\widehat r\supfk(\z)\right\}\right]=o_p(n^{-1/4}),
\end{align*}
which combined with Assumption \ref{asu:1} that $\dot g(\cdot)$ is Lipsitz, leads to
\begin{align*}
&\frac{K}{n(K-1)h^{p_{\z}}}\sum_{i\in \Isc\subminusk} K_h(\Z_i-\z){\omega}^*(\X_i)\bkappa\subismbbetan\dot g\left\{\bphi_i\trans\bar\bgamma+\bar r(\z)\right\}\left|\breve r\supfk(\z)-\widehat r\supfk(\z)\right|\\
=&o_p(n^{-1/4})+O_p\left([\widehat r\supfk(\z)-\bar r(\z)]^2+[\breve r\supfk(\z)-\bar r(\z)]^2\right).
\end{align*}
Using Assumption \ref{asu:1}(iv) and the weak law of large numbers, we can show that 
\[
\frac{K}{n(K-1)h^{p_{\z}}}\sum_{i\in \Isc\subminusk} K_h(\Z_i-\z){\omega}^*(\X_i)\bkappa\subismbbetan\dot g\left\{\bphi_i\trans\bar\bgamma+\bar r(\z)\right\}\asymp 1.
\]
Then by Lemma \ref{lemma:A3}, we conclude that $|\widehat r\supfk(\z)-\bar r(\z)|=o_p(1)$ uniformly for all $\z\in\mathcal{Z}$, and $\rE_1\{\widehat r\supfk(\Z)-\bar r(\Z)\}^2=o_p(n^{-1/2})$.

For $\widehat h\supfk(\cdot)$, we follow the same strategy to consider the difference between the second equation of (\ref{equ:kern:mom}) and equation (\ref{equ:kern:mom:ideal}), to derive that
\begin{align*}
&\frac{K}{n(K-1)h^{p_{\z}}}\sum_{i\in\Isc\subminusk}K_h(\Z_i-\z)\exp(\bpsi_i\trans\bar{\balpha})\bkappa\subismbbetan\breve{g}\{m^*(\X_i)\}\exp\{\bar h(\z)\}\left|\breve h\supfk(\z)-\widehat h\supfk(\z)\right|\\
=&O_p\left(\left[\rE_1\{\widetilde m\supfk(\X)-m^*(\X)\}^2\right]^{\frac{1}{2}}+\|\widetilde\bbeta\supfk-\bbeta_0\|_2\right)+O_p\left([\widehat h\supfk(\z)-\bar h(\z)]^2+[\breve h\supfk(\z)-\bar h(\z)]^2\right)\\
=&o_p(n^{-1/4})+O_p\left([\widehat h\supfk(\z)-\bar h(\z)]^2+[\breve h\supfk(\z)-\bar h(\z)]^2\right).
\end{align*}
Again combining this with Assumption \ref{asu:1}(iv) and Lemma \ref{lemma:A3}, we can derive that
\[
\sup_{\z\in\mathcal{Z}}|\widehat h\supfk(\z)-\bar h(\z)|=o_p(1);\quad\rE_1\{\widehat h\supfk(\Z)-\bar h(\Z)\}^2=o_p(n^{-1/2}).
\]
Thus we have finished proving Lemma \ref{lemma:A4}.

\end{proof}

\newpage

\section{Details of the extension discussed in Section \ref{sec:disc}}\label{sec:app:disc}

\subsection{Sieve estimator}\label{sec:app:sieve}

We consider $r(\Z)=\bxi\trans\b(\Z)$ and $h(\Z)=\bfeta\trans\b(\Z)$ where $\b(\Z)$ represents some prespecified basis function of $\Z$, e.g. natural spline or Hermite polynomials with diverging dimensionality, and $\bfeta$ and $\bxi$ represent their coefficients to estimate. In analog to (\ref{equ:kern:mom}), we propose to estimate the coefficients $\bxi$ and $\bfeta$ by solving
\begin{align*}
&\frac{K}{n(K-1)}\sum_{i\in \Isc\subminusk}\widetilde{\omega}\supfk(\X_i)\c\trans\bJhat_{\widetilde{\bbeta}\supfk}^{-1}\A_i\b(\Z_i)\left[Y_i-g\left\{\bphi_i\trans\widehat{\bgamma}\supfk+\bxi\trans \b(\Z_i)\right\}\right]=\bf{0};\\
&\frac{K}{n(K-1)}\sum_{i\in\Isc\subminusk}\c\trans\bJhat_{\widetilde{\bbeta}\supfk}^{-1}\A_i\breve{g}\{\widetilde m\supfk(\X_i)\} \exp\{\bpsi_i\trans\widehat{\balpha}\supfk+\bfeta\trans \b(\Z_i)\} \b(\Z_i)\\
=&\frac{1}{N}\sum_{i=n+1}^{n+N}\c\trans\bJhat_{\widetilde{\bbeta}\supfk}^{-1}\A_i\breve{g}\{\widetilde m\supfk(\X_i)\} \b(\Z_i).
\end{align*}
For one-dimensional $\Z_i$ occurring in our numerical studies, this sieve approach should have similar performance as kernel smoothing. While if $p_{\z}>1$ and $\Z_i=(Z_{i1},\ldots,Z_{ip_{\z}})\trans$, classic nonparametric approaches like kernel smoothing and sieve could have poor performance due to the curse of dimensionality. One may use additive model of $Z_{i1},\ldots,Z_{ip_{\z}}$ (constructed with the basis $\{\b\trans(Z_{i1}),\ldots,\b\trans(Z_{ip_{\z}})\}\trans$) instead of the fully nonparametric model for $\Z_i$, to avoid excessive model complexity. 

\subsection{General machine learning method}\label{sec:app:ml}

Given a response $A$, predictors $\bC$, and an arbitrary blackbox learning algorithm $\Lsc$, we let $\widehat \Esc^{\Lsc}[A\mid\bC]$ and $\widehat \Psc^{\Lsc}(A\mid\bC)$ denote the conditional expectation and conditional probability density (or mass) function of $A$ on $\bC$ estimated using the learning algorithm $\Lsc$. Here, we neglect the index of training samples in our notation for simplicity while in general, one should follow the established work like \cite{chernozhukov2016double}, to adopt cross-fitting, and ensure that $\widehat \Esc^{\Lsc}[A\mid\bC]$ and $\widehat \Psc^{\Lsc}(A\mid\bC)$ are estimated using training data independent with their plug-in samples.

Without loss of generality, we assume that knowing $\X$ is sufficient to identify $\Z$, $\bphi$ and $\bpsi$. We propose novel procedures using $\Lsc$ to estimate and calibrate the nuisance models. First, we regress $Y$ on $\X$ on $\Ssc$ using learning algorithm $\Lsc$ to obtain $\widehat\Esc^{\Lsc}[Y\mid\X]$, and regress $S$ on $\X$ to obtain $\widehat\Psc^{\Lsc}(S=1\mid\X)$. Also, we use $\Lsc$ to learn $\widehat\Psc^{\Lsc}(\X\mid\Z,S=1)$, i.e. the conditional distribution of $\X$ given $\Z$ on the source population. Then we solve:
\begin{equation}
\begin{split}
&\frac{K}{n(K-1)}\sum_{i\in \Isc\subminusk}\bphi_i\left\{\widehat\Esc^{\Lsc}[Y_i\mid\X_i]-g[\bphi_i\trans\bgamma+r(\Z_i)]\right\}=\bzero,\\
&\int_{\x\in\mathcal{X}\cap\{\z\}}\widehat\Psc^{\Lsc}(\x\mid\Z=\z,S=1)\left\{\widehat\Esc^{\Lsc}[Y\mid\X=\x]-g[\bphi_i\trans\bgamma+r(\z)]\right\}d\x=0,\quad\mbox{for }\z\in\mathcal{Z},    
\end{split}    
\label{equ:app:c1}
\end{equation}
to obtain the preliminary estimators $\widetilde\bgamma\supfk$ and $\widetilde r\supfk(\cdot)$, where $\x\in\mathcal{X}\cap\{\z\}$ represents the set of $\X$ belonging to its domain $\mathcal{X}$ and satisfying $\Z=\z$ for the fixed $\z$. To solve (\ref{equ:app:c1}) numerically, we adopt a monte carlo procedure introduced as follow. Let $M$ be some pre-specified number much larger than $n$, says $100n$. For each $i\in\Isc\supfk$, sample $\X_{i,1}$, $\X_{i,2}$,..., $\X_{i,M}$ independently from the estimated $\widehat\Psc^{\Lsc}(\X_i\mid\Z_i,S_i=1)$ given $\Z_{i,m}=\Z_i$ for each $m\in\{1,\ldots,M\}$. Then solve the estimating equation:
\begin{align*}
&\frac{K}{nM(K-1)}\sum_{i\in \Isc\subminusk}\sum_{m=1}^M\bphi_{i,m}\left\{\widehat\Esc^{\Lsc}[Y_{i,m}\mid\X_{i,m}]-g(\bphi_{i,m}\trans\bgamma+r_i)\right\}=\bzero,\\
&\frac{1}{M}\sum_{m=1}^M\widehat\Esc^{\Lsc}[Y_{i,m}\mid\X_{i,m}]-g(\bphi_{i,m}\trans\bgamma+r_i)=\bzero,\quad\mbox{for }i\in\Isc\supfk,
\end{align*}
to obtain the estimators $\widetilde\bgamma\supfk$ and $\widetilde r_i$, and set $\widetilde r\supfk(\Z_i)=\widetilde r_i$ for each $i\in\Isc\supfk$. Based on these estimators, we construct the debiased estimator for $\bgamma$ generally satisfying Assumption \ref{asu:3}(i). In specific, we use $\Lsc$ to obtain the estimators $\widehat\Esc^{\Lsc}[\bphi\dot g\{(\widetilde\bgamma\supfk)\trans\bphi+\widetilde r\supfk(\Z)\}|\Z,S=1]$ and $\widehat\Esc^{\Lsc}[g\{(\widetilde\bgamma\supfk)\trans\bphi+\widetilde r\supfk(\Z)\}|\Z,S=1]$. Then we let
\[
\widetilde\bdelta_i=(\widetilde\delta_{i1},\ldots,\widetilde\delta_{ip_{\bphi}})\trans=\bphi_i-\frac{\widehat\Esc^{\Lsc}[\bphi_i\dot g\{(\widetilde\bgamma\supfk)\trans\bphi_i+\widetilde r\supfk(\Z_i)\}|\Z_i,S_i=1]}{\widehat\Esc^{\Lsc}[g\{(\widetilde\bgamma\supfk)\trans\bphi_i+\widetilde r\supfk(\Z_i)\}|\Z_i,S_i=1]},
\]
solve 
\[
\widetilde\w_j\supfk=\min_{\w}\frac{K}{n(K-1)}\sum_{i\in\Isc_{\text{-}k}}\dot g\{(\widetilde\bgamma\supfk)\trans\bphi_i+\widetilde r\supfk(\Z_i)\}\left(\widetilde\delta_{ij}-\w\trans\widetilde\bdelta_{i,\text{-}j}\right)^2,
\]
for each $j\in\{1,\ldots,p_{\bphi}\}$, and let $\widetilde\bepsilon_i=(\widetilde\epsilon_{i1},\ldots,\widetilde\epsilon_{ip_{\bphi}})\trans$, where $\widetilde\epsilon_{ij}=\widetilde\delta_{ij}-(\widetilde\w_j\supfk)\trans\widetilde\bdelta_{i,\text{-}j}$, and
\[
\widetilde\sigma_j^2=\frac{K}{n(K-1)}\sum_{i\in\Isc_{\text{-}k}}\widetilde\epsilon_{ij}^2\dot g\left\{(\widetilde\bgamma\supfk)\trans\bphi_i+\widetilde r\supfk(\Z_i)\right\}.
\]
Then we construct the debiased estimator  $\bgammahat\supfk=(\widehat\gamma\supfk_1,\ldots,\widehat\gamma\supfk_{p_{\bphi}})\trans$ through:
\begin{equation}
\widehat\gamma\supfk_j=\widetilde\gamma\supfk_j+\frac{K}{n(K-1)}\sum_{i\in\Isc_{\text{-}k}}\frac{\widetilde\epsilon_{ij}}{\widetilde\sigma_j}\left[Y_i-g\{(\widetilde\bgamma\supfk)\trans\bphi_i+\widetilde r\supfk(\Z_i)\}\right].
\label{equ:app:c2}
\end{equation}
Finally, the calibrated estimator of the nuisance component $r(\cdot)$ is obtained by solving $\widehat r_i$ from:
\[
\frac{1}{M}\sum_{m=1}^M\widetilde{\omega}\supfk(\X_{i,m})\c\trans\bJhat_{\widetilde{\bbeta}\supfk}^{-1}\A_{i,m}\left[\widehat\Esc^{\Lsc}[Y_{i,m}\mid\X_{i,m}]-g\left\{\bphi_{i,M}\trans\widehat{\bgamma}\supfk+r_i\right\}\right]=0,
\]
for each $i$, and set $\widehat r\supfk(\Z_i)=\widehat r_i$, where $\widetilde{\bbeta}\supfk$ is again solved through:
\[
\frac{K}{n(K-1)}\sum_{i\in\Isc\subminusk}\widetilde{\omega}\supfk(\X_i)\A_i\{Y_i-\widetilde{m}\supfk(\X_i)\}+\frac{1}{N}\sum_{i=n+1}^{N+n}\A_i\{\widetilde{m}\supfk(\X_i)-g(\A_i\trans\bbeta)\}=\mathbf{0}.
\]
Noting that our above introduced procedure is applicable to any semi-non-parametric M-estimation problem, so the preliminary estimator $\widetilde{\omega}\supfk(\X_i)$ and the calibrated estimator for $\balpha$ and $h(\cdot)$ can be obtained in the same way. 

\begin{remark}
Our construction procedure proposed in this section involves estimation of the probability density function, which is typically more challenging than purely estimating the conditional mean for a machine learning method. Note that for linear, log-linear and logistic model, one can avoid estimating probability density function to construct the doubly robust (double machine learning) estimators; see \cite{dukes2020inference,ghosh2020doubly,liu2020note}. Thus, when the link function $g(a)=a$, $g(a)=e^a$ or $g(a)=e^a/(1+e^a)$, our construction actually does not require estimating the probability density function with $\Lsc.$
\end{remark}

At last, we provide discussion and justification towards the $n^{1/2}$-consistency and asymptotic normality of the debiased estimator $\widehat\bgamma\supfk$. In specific, we take $\bar\bgamma=\bgamma^*$, and write (\ref{equ:app:c2}) as:
\begin{align*}
\widehat\gamma\supfk_j=&\widetilde\gamma\supfk_j+\frac{K}{n(K-1)}\sum_{i\in\Isc_{\text{-}k}}\frac{\widetilde\epsilon_{ij}}{\widetilde\sigma_j}\Bigg[Y_i-\rE_1[Y_i\mid\X_i]+\rE_1[Y_i\mid\X_i]-g\{(\bgamma^*)\trans\bphi_i+r^*(\Z_i)\}\\
&\hspace{5.5cm}+g\{\bar\bgamma\trans\bphi_i+r^*(\Z_i)\}-g\{(\widetilde\bgamma\supfk)\trans\bphi_i+\widetilde r\supfk(\Z_i)\}\Bigg].
\end{align*}
Note that $Y_i-\rE_1[Y_i\mid\X_i]$ is orthogonal to $\widetilde\epsilon_{ij}$ and its estimation error since the latter is deterministic on $\X_i$. According to our moment equation for $\bgamma^*$ and $r^*(\cdot)$, $\rE_1[Y_i\mid\X_i]-g\{(\bgamma^*)\trans\bphi_i+r^*(\Z_i)\}$ is orthogonal to arbitrary (regular) function of $\Z_i$ and linear function of $\bphi_i$, so is also orthogonal to $\widetilde\epsilon_{ij}$ and its estimation error. In addition, by our construction,
\[
\rE_1\left(\bphi_i-\frac{\rE_1[\bphi_i\dot g\{(\bgamma^*)\trans\bphi_i+r^*(\Z_i)\}\mid\Z_i]}{\rE_1[\dot g\{(\bgamma^*)\trans\bphi_i+r^*(\Z_i)\}\mid\Z_i]}\right)=\bzero,
\]
and $\widetilde\epsilon_{ij}$ is orthogonal to any linear function of $\bphi_{i,\text{-}j}$ and $\bdelta_{i,\text{-}j}$. So the first order error in $g\{\bar\bgamma\trans\bphi_i+r^*(\Z_i)\}-g\{(\widetilde\bgamma\supfk)\trans\bphi_i+\widetilde r\supfk(\Z_i)\}$, i.e. $\dot g\{\bar\bgamma\trans\bphi_i+r^*(\Z_i)\}\{(\widetilde\bgamma\supfk-\bar\bgamma)\trans\bphi_i+r^*(\Z_i)-\widetilde r\supfk(\Z_i)\}$, is orthogonal to $\widetilde\epsilon_{ij}$ for each $j$. Thus, all the first order error terms in $\widehat\gamma\supfk_j-\bar\bgamma$ could be removed through our Neyman orthogonal construction. 

Inspired by existing work of double machine learning like \cite{chernozhukov2018double} and \cite{liu2020note}, when the mean squared error of machine learning algorithm $\Lsc$ has the convergence rates $o_p(n^{-1/2})$ with respect to all the learning objectives included in this section, i.e. the rate double robustness property, the machine learning estimator $\widehat r\supfk(\cdot)$ satisfies Assumption \ref{asu:3}(ii). Also, the second order error of $\widehat\gamma\supfk_j-\bar\bgamma$ could be removed asymptotically. And consequently, $\widehat\bgamma\supfk$ satisfy Assumption \ref{asu:3}(i). Again, these arguments are applicable to the nuisance estimators for $\balpha$ and $h(\cdot)$ derived in the same way. Therefore, our proposed nuisance estimators introduced in this section tend to satisfy Assumption \ref{asu:3}.

\subsection{Intrinsic efficient construction}\label{sec:app:disc:detail:intr}
 
In this section, we introduce the intrinsic efficient construction of the imputation model under our framework. For simplicity, we consider a semi-supervised setting with $n$ labeled source samples and $N\gg n$ unlabeled target samples. The augmentation approach proposed by \cite{shu2018improved} could be used for extending our method to the $N\asymp n$ case. For some given $h(\cdot)$, let the estimating equation of $\widetilde\balpha\supfk$ be
\[
\sum_{i\in\{n+1,\ldots,n+N\}\cup\Isc\subminusk}\S\{\delta_i,\X_i;\balpha,h(\cdot)\}=\mathbf{0},
\]
with $\S\{\delta_i,\X_i;\balpha,h(\cdot)\}$ representing the score function. For example, one can take
\[
\S\{\delta_i,\X_i;\balpha,h(\cdot)\}=\delta_i\exp\{\bpsi_i\trans\balpha+h(\Z_i)\}\bpsi_i-|\Isc\subminusk|(1-\delta_i)\bpsi_i/N.
\]
Denote that $\S_i=\S\{\delta_i,\X_i;\widetilde\balpha\supfk,\widetilde h\supfk(\cdot)\}$ and let $\bPi_{\Isc\subminusk}(\epsilon_i;\S_i)$ be the empirical projection operator of any variable $\epsilon_i$ to the space spanned by $\S_i$ on the samples $\Isc\subminusk$
and $\bPi^{\perp}_{\Isc\subminusk}(\epsilon_i;\S_i)=\epsilon_i-\bPi_{\Isc\subminusk}(\epsilon_i;\S_i)$. When the importance weight model is correctly specified and $N\gg n$, the empirical asymptotic variance for $\c\trans\bbetahat\subatrel$ with nuisance parameters $\bgamma$ and $r(\cdot)$ can be expressed as
\begin{equation}
\frac{K}{n(K-1)}\sum_{i\in \Isc\subminusk}\left[\widetilde{\omega}\supfk(\X_i)\bPi^{\perp}_{\Isc\subminusk}\left(\c\trans\bJhat_{\widetilde{\bbeta}\supfk}^{-1}\A_i[Y_i-g\{\bphi_i\trans\bgamma+r(\Z_i)\}];\S_i\right)\right]^2.
\label{equ:var}
\end{equation}
Then the intrinsically efficient construction of the imputation model is given by minimizing (\ref{equ:var}) subject to the moment constraint:
\[
\frac{1}{|\Isc\subminusk\cap \Isc^a|}\sum_{i\in \Isc\subminusk\cap \Isc^a} K_h(\Z_i-\z)\widetilde{\omega}\supfk(\X_i)\c\trans\bJhat_{\widetilde{\bbeta}\supfk}^{-1}\A_i\left[Y_i-g\left\{\bphi_i\trans\bgamma+r(\Z)\right\}\right]=0,
\]
which is the same as the first equation of (\ref{equ:kern:mom}) except that both $\bgamma$ and $r(\Z)$ are unknown here. This optimization problem could be solved with methods like profile kernel and back-fitting \citep{lin2006semiparametric}. Alternatively and more conveniently, one could use sieve, as discussed in Appendix~\ref{sec:app:sieve}, to model $r(\Z_i)$ and use a constrained least square regression: let $\b(\Z)$ be some basis function of $\z$ and solve
\begin{align*}
&\min_{\bgamma,\bxi}~\sum_{i\in \Isc\subminusk}\left[\widetilde{\omega}\supfk(\X_i)\bPi^{\perp}_{\Isc\subminusk}\left(\c\trans\bJhat_{\widetilde{\bbeta}\supfk}^{-1}\A_i[Y_i-g\{\bphi_i\trans\bgamma+\b\trans(\Z_i)\bxi\}];\S_i\right)\right]^2;\\
&{\rm s.t.}~\sum_{i\in \Isc\subminusk\cap \Isc^a} \b(\Z_i)\widetilde{\omega}\supfk(\X_i)\c\trans\bJhat_{\widetilde{\bbeta}\supfk}^{-1}\A_i\left[Y_i-g\left\{\bphi_i\trans\bgamma+\b\trans(\Z_i)\bxi\right\}\right]=0,
\end{align*}
to obtain $\widetilde\bgamma\supfk$ and $\widetilde r\supfk(\Z)=\b\trans(\Z)\widetilde\bxi\supfk$ simultaneously. 
To get the intrinsic efficient estimator for a nonlinear but differentiable function $\ell(\bbeta_0)$, with its gradient being $\dot{\ell}(\cdot)$, we first estimate the entries $\beta_{0i}$ using our proposed method for every $i\in\{1,2,\ldots,d\}$ and use them to form a preliminary $\sqrt{n}$-consistent estimator $\bbetahat_{(init)}$. Then we estimate the linear function $\bbeta_0\trans\dot{\ell}\{\bbetahat_{(init)}\}$ with the intrinsically efficient estimator and utilize the expansion $\ell(\bbeta_0)\approx\ell\{\bbetahat_{(init)}\}+\{\bbeta_0-\bbetahat_{(init)}\}\trans\dot{\ell}\{\bbetahat_{(init)}\}$ for an one-step update.


\newpage

\section{Implementing details and additional results of simulation}\label{sec:app:results:sim}

To obtain the preliminary estimators $\widetilde \omega\supfk(\cdot)$ and $\widetilde m\supfk(\cdot)$ of our method, we use semiparametric logistic regression with covariates including the parametric basis and the natural splines of the nonparametric components $Z$ with order $[n^{1/4}]$ for the imputation model and $[(N+n)^{1/4}]$ for the importance weight model. In this process, we add ridge penalty tuned by cross-validation with tuning parameter of order $n^{-2/3}$ (below the parametric rate) to enhance the training stability. 

We set the loading vector $\c$ as $(1,0,0,0)\trans$, $(0,1,0,0)\trans$, $(0,0,1,0)\trans$, and $(0,0,0,1)\trans$ to estimate $\beta_0,\beta_1,\beta_2,\beta_3$ separately. For $\beta_1,\beta_2,\beta_3$, the weights $\c\trans\bJhat_{\widetilde{\bbeta}\supfk}^{-1}\A_i$'s are not positive definite so we split the source and target samples as $\Isc^{+}=\{i:\c\trans\bJhat_{\widetilde{\bbeta}\supfk}^{-1}\A_i\geq 0\}$ and $\Isc^{-}=\{i:\c\trans\bJhat_{\widetilde{\bbeta}\supfk}^{-1}\A_i< 0\}$ as introduced in Remark \ref{rem:sign:split}, and use (\ref{equ:kern:mom:sign}) to estimate their nonparametric components. For $\beta_0$, we find that $\c\trans\bJhat_{\widetilde{\bbeta}\supfk}^{-1}\A_i$ is nearly positive definite under all configurations but these weights are sometimes of high variation. So we also split the source/target samples by cutting the $\c\trans\bJhat_{\widetilde{\bbeta}\supfk}^{-1}\A_i$'s with their median, to reduce the variance of weights at each fold and improve the effective sample size. We use cross-fitting with $K=5$ folds for our method and the two double machine learning estimators. And all the tuning parameters including the bandwidth of our method and kernel machine and the coefficients of the penalty functions are selected by 5-folded cross-validation on the training samples. We present the estimation performance (mean square error, bias and coverage probability) on each parameter in Tables~\ref{tab:simu:IMPcornp}--\ref{tab:simu:IWcornonp}, for the four configurations separately. 

\begin{table}[!htbp] \centering 
\caption{\label{tab:simu:IMPcornp} Estimation performance of the methods on parameters $\beta_0,\beta_1,\beta_2,\beta_3$ under Configuration (i) described in Section \ref{sec:simu}. Parametric: doubly robust estimator with parametric nuisance models; ATReL: our proposed doubly robust estimator using semi-non-parametric nuisance models; 
DML$\subBE$: double machine learning with flexible basis expansions; DML$\subKM$: double machine learning with kernel machine. RMSE: root mean square error; CP: coverage probability of the $95\%$ confidence interval.} 
\begin{tabular}{@{\extracolsep{5pt}} cccccc} 
\\[-1.8ex]\hline 
\hline \\[-1.8ex] 
& &\multicolumn{4}{c}{Estimator}  \\ \cline{3-6}
 \\[-2ex]
Covariates & & Parametric  & ATReL & DML$\subBE$ & DML$\subKM$\\
\hline \\[-1.8ex]
\multirow{3}*{$\beta_0$}&RMSE & $0.102$ & $0.110$ & $0.168$ & $0.116$ \\ 
&Bias & $-0.007$ & $0.0005$ & $0.112$ & $0.010$ \\ 
&CP & $0.95$ & $0.95$ & $0.84$ & $0.93$ \\ \\[-1.8ex]\cline{2-6} \\[-1.8ex]
\multirow{3}*{$\beta_1$}&RMSE & $0.181$ & $0.124$ & $0.160$ & $0.198$ \\ 
&Bias & $-0.146$ & $-0.056$ & $-0.104$ & $-0.163$ \\ 
&CP & $0.91$ & $0.93$ & $0.92$ & $0.85$ \\ \\[-1.8ex]\cline{2-6} \\[-1.8ex]
\multirow{3}*{$\beta_2$}&RMSE & $0.133$ & $0.126$ & $0.191$ & $0.134$ \\ 
&Bias & $0.059$ & $0.032$ & $-0.109$ & $-0.017$ \\ 
&CP & $0.99$ & $0.97$ & $0.94$ & $0.98$ \\ \\[-1.8ex]\cline{2-6} \\[-1.8ex]
\multirow{3}*{$\beta_3$}&RMSE & $0.137$ & $0.133$ & $0.195$ & $0.150$ \\ 
&Bias & $0.049$ & $0.030$ & $-0.108$ & $-0.040$ \\ 
&CP & $0.99$ & $0.97$ & $0.96$ & $0.97$ \\ 
\hline\hline \\[-1.8ex] 
\end{tabular} 
\end{table}

\begin{table}[!htbp] \centering 
\caption{\label{tab:simu:IMPcornonp} Estimation performance of the methods on parameters $\beta_0,\beta_1,\beta_2,\beta_3$ under Configuration (ii) described in Section \ref{sec:simu}. Parametric: doubly robust estimator with parametric nuisance models; ATReL: our proposed doubly robust estimator using semi-non-parametric nuisance models; 
DML$\subBE$: double machine learning with flexible basis expansions; DML$\subKM$: double machine learning with kernel machine. RMSE: root mean square error; CP: coverage probability of the $95\%$ confidence interval.} 
\begin{tabular}{@{\extracolsep{5pt}} cccccc} 
\\[-1.8ex]\hline 
\hline \\[-1.8ex] 
& &\multicolumn{4}{c}{Estimator}  \\ \cline{3-6}
 \\[-2ex]
Covariates & & Parametric & ATReL & DML$\subBE$ &  DML$\subKM$\\
\hline \\[-1.8ex]
\multirow{3}*{$\beta_0$}&RMSE & $0.108$ & $0.114$ & $0.186$ & $0.124$ \\ 
&Bias & $-0.004$ & $0.004$ & $0.136$ & $0.018$ \\ 
&CP & $0.92$ & $0.94$ & $0.82$ & $0.90$ \\ \\[-1.8ex]\cline{2-6}\\[-1.8ex] 
\multirow{3}*{$\beta_1$}&RMSE & $0.107$ & $0.118$ & $0.144$ & $0.122$ \\ 
&Bias & $-0.001$ & $-0.015$ & $-0.062$ & $-0.046$ \\ 
&CP & $0.99$ & $0.95$ & $0.95$ & $0.98$ \\ \\[-1.8ex]\cline{2-6}\\[-1.8ex] 
\multirow{3}*{$\beta_2$}&RMSE & $0.129$ & $0.131$ & $0.209$ & $0.166$ \\ 
&Bias & $-0.006$ & $-0.024$ & $-0.136$ & $-0.084$ \\ 
&CP & $0.98$ & $0.96$ & $0.94$ & $0.95$ \\ \\[-1.8ex]\cline{2-6}\\[-1.8ex] 
\multirow{3}*{$\beta_3$}&RMSE & $0.124$ & $0.128$ & $0.200$ & $0.171$ \\ 
&Bias & $-0.008$ & $-0.019$ & $-0.123$ & $-0.097$ \\ 
&CP & $0.98$ & $0.97$ & $0.94$ & $0.96$ \\ 
\hline\hline \\[-1.8ex]
\end{tabular} 
\end{table}

\begin{table}[!htbp] \centering 
\caption{\label{tab:simu:IWcornp}Estimation performance of the methods on parameters $\beta_0,\beta_1,\beta_2,\beta_3$ under Configuration (iii) described in Section \ref{sec:simu}. Parametric: doubly robust estimator with parametric nuisance models; ATReL: our proposed doubly robust estimator using semi-non-parametric nuisance models; 
DML$\subBE$: double machine learning with flexible basis expansions; DML$\subKM$: double machine learning with kernel machine. RMSE: root mean square error; CP: coverage probability of the $95\%$ confidence interval.} 
\begin{tabular}{@{\extracolsep{5pt}} cccccc} 
\\[-1.8ex]\hline 
\hline \\[-1.8ex] 
& &\multicolumn{4}{c}{Estimator}  \\ \cline{3-6}
 \\[-2ex]
Covariates & & Parametric & ATReL & DML$\subBE$ &  DML$\subKM$\\
\hline \\[-1.8ex]
\multirow{3}*{$\beta_0$}&RMSE & $0.113$ & $0.112$ & $0.134$ & $0.114$ \\ 
&Bias & $-0.052$ & $-0.014$ & $-0.064$ & $-0.026$ \\ 
&CP & $0.93$ & $0.95$ & $0.93$ & $0.95$ \\ \\[-1.8ex]\cline{2-6}\\[-1.8ex] 
\multirow{3}*{$\beta_1$}&RMSE & $0.341$ & $0.151$ & $0.152$ & $0.189$ \\ 
&Bias & $-0.300$ & $-0.047$ & $-0.043$ & $-0.135$ \\ 
&CP & $0.82$ & $0.93$ & $0.95$ & $0.86$ \\ \\[-1.8ex]\cline{2-6}\\[-1.8ex] 
\multirow{3}*{$\beta_2$}&RMSE & $0.145$ & $0.133$ & $0.141$ & $0.133$ \\ 
&Bias & $-0.006$ & $-0.011$ & $-0.035$ & $-0.054$ \\ 
&CP & $0.95$ & $0.94$ & $0.95$ & $0.91$ \\ \\[-1.8ex]\cline{2-6}\\[-1.8ex] 
\multirow{3}*{$\beta_3$}&RMSE & $0.143$ & $0.137$ & $0.139$ & $0.131$ \\ 
&Bias & $-0.008$ & $0.004$ & $0.003$ & $-0.033$ \\ 
&CP & $0.94$ & $0.95$ & $0.95$ & $0.91$ \\ 
\hline\hline \\[-1.8ex] 
\end{tabular} 
\end{table}

\begin{table}[!htbp] \centering 
\caption{\label{tab:simu:IWcornonp} Estimation performance of the methods on parameters $\beta_0,\beta_1,\beta_2,\beta_3$ under Configuration (iv) described in Section \ref{sec:simu}. Parametric: doubly robust estimator with parametric nuisance models; ATReL: our proposed doubly robust estimator using semi-non-parametric nuisance models; 
DML$\subBE$: double machine learning with flexible basis expansions; DML$\subKM$: double machine learning with kernel machine. RMSE: root mean square error; CP: coverage probability of the $95\%$ confidence interval.} 
\begin{tabular}{@{\extracolsep{5pt}} cccccc} 
\\[-1.8ex]\hline 
\hline \\[-1.8ex] 
& &\multicolumn{4}{c}{Estimator}  \\ \cline{3-6}
 \\[-2ex]
Covariates & & Parametric & ATReL & DML$\subBE$ &  DML$\subKM$\\
\hline \\[-1.8ex]
\multirow{3}*{$\beta_0$}&RMSE & $0.103$ & $0.107$ & $0.189$ & $0.109$ \\ 
&Bias & $-0.003$ & $0.010$ & $0.151$ & $0.027$ \\ 
&CP & $0.95$ & $0.95$ & $0.73$ & $0.95$ \\ \\[-1.8ex]\cline{2-6}\\[-1.8ex] 
\multirow{3}*{$\beta_1$}&RMSE & $0.140$ & $0.128$ & $0.132$ & $0.156$ \\ 
&Bias & $-0.008$ & $0.008$ & $0.035$ & $0.100$ \\ 
&CP & $0.94$ & $0.93$ & $0.94$ & $0.86$ \\ \\[-1.8ex]\cline{2-6}\\[-1.8ex] 
\multirow{3}*{$\beta_2$}&RMSE & $0.137$ & $0.126$ & $0.127$ & $0.121$ \\ 
&Bias & $-0.004$ & $-0.004$ & $-0.025$ & $0.000$ \\ 
&CP & $0.96$ & $0.96$ & $0.95$ & $0.90$ \\ \\[-1.8ex]\cline{2-6}\\[-1.8ex] 
\multirow{3}*{$\beta_3$}&RMSE & $0.139$ & $0.126$ & $0.121$ & $0.122$ \\ 
&Bias & $0.005$ & $0.015$ & $0.022$ & $0.050$ \\ 
&CP & $0.95$ & $0.97$ & $0.96$ & $0.93$ \\ 
\hline\hline \\[-1.8ex] 
\end{tabular} 
\end{table}

\newpage

\section{Implementing details and additional results of real example}\label{sec:app:results:real}

The specific nuisance model constructions are described as follows.

\begin{table}[!htb]
\begin{tabular}{p{4.5cm}|p{5.5cm}|p{5.5cm}} 
\hline
\hline
Method  & Importance weighting & Imputation \\ \hline
Parametric & Logistic model with $\bPsi=(\X\trans,X_1X_2,X_1X_3,X_2X_3) \trans$ & Logistic model with $\bPhi=\X$ \\ \hline
{ATReL} (our method) & Logistic model with $\bPsi=(\X\trans,X_1X_2,X_1X_3,X_2X_3) \trans$ and set $Z=X_2$ for nonparametric modeling & Logistic model with $\bPhi=\X$ and set $Z=X_2$ for nonparametric modeling\\ \hline
Double machine learning with flexible basis expansions &   $\ell_1 + \ell_2$ regularized regression including basis terms: $\X$, natural splines of $X_1$, $X_2$ and $X_6$ of order $5$ and interaction terms of these natural splines &  $\ell_1 + \ell_2$ regularized regression including basis terms: $\X$, natural splines of $X_1$, $X_2$ and $X_6$ of order $5$ and interaction terms of these natural splines \\ \hline
Double machine learning with kernel machine & Support vector machine with the radial basis function kernel & Support vector machine with the radial basis function kernel \\ \hline
\end{tabular}
\end{table}




\noindent We present the fitted coefficients of all the included approaches in Table~\ref{tab:real:beta2015}.
\begin{table}[H]
\centering
\caption{\label{tab:real:beta2015} Estimators of the target model coefficients. $\beta_0,\beta_1,\beta_2,\beta_3,\beta_4$ represent respectively the intercept, coefficient of the total healthcare utilization ($X_1$), coefficient of the log(NLP+1) of RA ($X_2$), coefficient of the indicator for NLP mention of tumor necrosis factor (TNF) inhibitor ($X_3$), and coefficient of the indicator for NLP mention of bone erosion ($X_4$). Parametric: doubly robust estimator with parametric nuisance models; ATReL: our proposed doubly robust estimator using semi-non-parametric nuisance models; 
DML$\subBE$: double machine learning with flexible basis expansions; DML$\subKM$: double machine learning with kernel machine.}
\begin{tabular}{ccccccc}
\hline
\hline
          & Source & Parametric & {ATReL} & DML$\subBE$ & DML$\subKM$ & {Target} \\ \hline
$\beta_0$ &   -5.70 & -5.08 & -5.75 & -8.88 & -5.73 & -5.03 \\   
$\beta_1$  & 0.03 & 0.12 & -0.19 & 0.01 & 0.05 & -0.31 \\ 
$\beta_2$     & 1.73 & 1.39 & 1.56 & 2.64 & 1.61 & 1.35 \\  
$\beta_3$  & 0.69 & 0.62 & 0.78 & 0.77 & 0.66 & 0.94 \\ 
$\beta_4$ & 0.60 & 0.62 & 0.44 & 0.62 & 0.35 & 0.14 \\ 
\hline\hline\\[-1.8ex] 
\end{tabular}
\end{table}

\end{document}